\documentclass[%
 aip,
 reprint,
 citeautoscript,
 superscriptaddress,
]{revtex4-1}

\usepackage{amsthm,amsmath,amssymb}
\usepackage{cleveref}
\usepackage{graphicx} 
\usepackage{dcolumn} 
\usepackage{bm} 
\usepackage{romanbar}
\usepackage{subcaption}

\usepackage[utf8]{inputenc}
\usepackage[T1]{fontenc}
\usepackage{soul}
\usepackage{siunitx}
\usepackage{xcolor} 
\usepackage{comment} 

\newtheorem{theorem}{Theorem}
\newtheorem{corollary}[theorem]{Corollary}

\renewcommand{\vec}[1]{\boldsymbol{#1}} 
\newcommand{\m}[1]{\boldsymbol{#1}} 
\newcommand{\beq}{\begin{equation}}
\newcommand{\eeq}{\end{equation}}
\newcommand{\bea}{\begin{eqnarray}}
\newcommand{\eea}{\end{eqnarray}}

\DeclareMathOperator{\tr}{tr}

\DeclareMathOperator{\Exp}{exp}

\def\E{\mathbb{E}}

\mathchardef\ordinarycolon\mathcode`\:
\mathcode`\:=\string"8000
\begingroup \catcode`\:=\active
  \gdef:{\mathrel{\mathop\ordinarycolon}}
\endgroup

\DeclareMathOperator{\sech}{sech}

\begin{document}

\title{A generalized class of strongly stable and dimension-free T-RPMD integrators}

\author{Jorge L. Rosa-Ra\'ices$^\ast$}
\affiliation{%
Division of Chemistry and Chemical Engineering, California Institute of Technology, Pasadena, CA 91125, USA%
}
\author{Jiace Sun$^\ast$}
\affiliation{%
Division of Chemistry and Chemical Engineering, California Institute of Technology, Pasadena, CA 91125, USA%
}
\author{Nawaf Bou-Rabee}
\email{nawaf.bourabee@rutgers.edu}
\affiliation{%
Department of Mathematical Sciences, Rutgers University Camden, Camden, NJ 08102, USA%
}
\author{Thomas F. Miller III}
\email{tfm@caltech.edu}
\affiliation{%
Division of Chemistry and Chemical Engineering, California Institute of Technology, Pasadena, CA 91125, USA%
}

\begin{abstract}

Recent work shows that strong stability and dimensionality freedom are essential for robust numerical integration of thermostatted ring-polymer molecular dynamics (T-RPMD) and path-integral molecular dynamics (PIMD), without which standard integrators exhibit non-ergodicity and other pathologies [J.\ Chem.\ Phys.\ \textbf{151}, 124103 (2019); J.\ Chem.\ Phys.\ \textbf{152}, 104102 (2020)].
In particular, the BCOCB scheme, obtained via Cayley modification of the standard BAOAB scheme, features a simple reparametrization of the free ring-polymer sub-step that confers strong stability and dimensionality freedom and has been shown to yield excellent numerical accuracy in condensed-phase systems with large time-steps. 
Here, we introduce a broader class of T-RPMD numerical integrators that exhibit strong stability and dimensionality freedom, irrespective of the Ornstein--Uhlenbeck friction schedule.
In addition to considering equilibrium accuracy and time-step stability as in previous work, we evaluate the integrators on the basis of their rates of convergence to equilibrium and their efficiency at evaluating equilibrium expectation values.
Within the generalized class, we find BCOCB to be superior with respect to accuracy and efficiency for various configuration-dependent observables, although other integrators within the generalized class perform better for velocity-dependent quantities.
Extensive numerical evidence indicates that the stated performance guarantees hold for the strongly anharmonic case of liquid water.
Both analytical and numerical results indicate that  BCOCB  excels over other known integrators in terms of accuracy, efficiency, and stability with respect to time-step for practical applications.

\end{abstract}

\maketitle

\section{Introduction}
\label{introduction}

Path-Integral Molecular Dynamics (PIMD) provides a practical and popular tool to simulate condensed-phase systems subject to strong nuclear quantum effects.\cite{Parrinello1984,Habershon2013,Markland2018}
Based on the ring-polymer correspondence between quantum and classical Boltzmann statistics,\cite{Feynman1965,Chandler1981} PIMD exploits the computational methods of molecular dynamics\cite{Frenkel2002,Rapaport2004,Leimkuhler2015,Allen2017} to approximate quantum thermodynamics and kinetics through various classical models.\cite{Cao1994,Craig2004,Liu2014,Hele2015,Hele2015a,Cendagorta2018}
Applications of PIMD include calculations of chemical reaction rates,\cite{Craig2005,Craig2005a} diffusion coefficients,\cite{Miller2005,Miller2005a} absorption spectra,\cite{Habershon2008,Witt2009} solid and liquid structure,\cite{Morrone2009,Cheng2019} and equilibrium isotope effects.\cite{Zimmermann2009,Eldridge2019}

Many numerical integration schemes for PIMD are based on a symmetric Trotter (i.e., Strang) splitting\cite{Trotter1959,Strang1968} of the exact time-evolution operator, and feature a sub-step for free ring-polymer propagation.\cite{Tuckerman1993,Ceriotti2010,Liu2016}
Due to fast harmonic motions present in the free ring polymer, a \textit{strongly stable} implementation of this sub-step is essential.\cite{Calvo2009,Arnol2013}
Strong stability can be achieved by one of two approaches.
The first approach introduces a preconditioned form of the equations of motion by modifying the ring-polymer mass matrix.
Preconditioning improves the stability of the exact free ring-polymer update at the expense of consistent dynamics.\cite{Tuckerman1993,Minary2003,Liu2016,BouRabee2018,BouRabee2019,Lu2020}
The second approach does not modify the ring-polymer mass matrix, leaving the dynamics non-preconditioned,\cite{Ceriotti2010,Ceriotti2011,Rossi2014,Zhang2017,Rossi2018} and instead replaces the exact free ring-polymer update with a strongly stable approximation.\cite{Korol2019}
We apply the latter approach in the current work to Thermostatted Ring-Polymer Molecular Dynamics (T-RPMD),\cite{Rossi2014} a non-preconditioned variant of PIMD featuring an Ornstein--Uhlenbeck thermostat that approximately preserves the real-time dynamical accuracy of RPMD for quantum correlation functions of a wide range of observables.\cite{Braams2006}

In addition to strong stability of the free ring-polymer update, another basic requirement of a numerical integrator for T-RPMD is non-zero overlap between the numerically sampled and exact ring-polymer configurational distributions in the limit of an infinite number of ring-polymer beads.
Standard integrators fail to satisfy this requirement at any finite integration time-step,\cite{Korol2020} which motivates the introduction of \textit{dimension-free} T-RPMD schemes that allow for accurate configurational sampling with large time-stepping and arbitrarily many ring-polymer beads.
We recently found that standard integrators could be made dimension-free through the introduction of a suitable strongly stable ring-polymer update,\cite{Korol2020} and the current paper investigates this finding in much greater generality.

To this end, we introduce a function $\theta$ that defines the free ring-polymer update and deduce how the choice of $\theta$ impacts the properties and performance of the corresponding T-RPMD integrator.
The case $\theta(x) = x$, i.e., $\theta$ is the identity, corresponds to the exact free ring-polymer update.
Therefore, to ensure second-order accuracy, $\theta$ must approximate the identity near the origin, i.e., $\theta(0) = 0$, $\theta'(0) = 1$ and $\theta''(0) = 0$.
Moreover, strong stability requires that the range of the function $\theta$ is within $(0, \pi)$ for $x > 0$, and ergodicity and dimensionality freedom of the corresponding T-RPMD integrator impose additional requirements on $\theta$.
There are many choices of $\theta$ that fulfill the identified requirements, including $\theta(x) = 2\arctan(x/2)$ which leads to the BCOCB scheme introduced in Ref.~\onlinecite{Korol2020}.
In fact, we find that this choice of $\theta$ is superior for the estimation of configurational averages via T-RPMD from the perspectives of accuracy and efficiency, despite its poor performance with respect to the ring-polymer velocities.

The paper is organized as follows.
In Section~\ref{theory} we recall exact T-RPMD and its time discretization, present the new function $\theta$ that determines the free ring-polymer update, and obtain sufficient conditions on $\theta$ to guarantee strong stability and dimensionality freedom of the corresponding T-RPMD integrator.
In Section~\ref{numerical_results}, we compare the performance of various $\theta$ in applications to the one-dimensional quantum harmonic oscillator and to a quantum-mechanical model of room-temperature liquid water.
Section~\ref{summary} summarizes the work, and Section~\ref{suppmat} provides supporting mathematical proofs and computational protocols.

\section{Theory}
\label{theory}

\subsection{T-RPMD}
\label{trpmd}

Consider a one-dimensional quantum particle with the Hamiltonian operator
\begin{equation}
    \hat{H} = \frac{1}{2m} \hat{p}^2 + V(\hat{q}) \;,
\end{equation}
where $m$ is the particle mass, $\hat{q}$ and $\hat{p}$ the position and momentum operators, and $V(\hat{q})$ a potential energy surface.
Ignoring exchange statistics, the properties of this system at thermal equilibrium are encoded in the quantum partition function
\begin{equation}
    Q = \tr[e^{-\beta \hat{H}}] \;, 
\end{equation}
where $\beta = (k_B T)^{-1}$, $k_B$ is the Boltzmann constant and $T$ the physical temperature.
Using a path-integral discretization (i.e., a Trotter factorization of the Boltzmann operator\cite{Trotter1959}), $Q = \lim_{n \to \infty} Q_n$ can be approximated by the classical partition function $Q_n$ of a ring polymer with $n$ beads,\cite{Feynman1965,Chandler1981}
\begin{equation}
    Q_n = \frac{m^n}{(2\pi\hbar)^n} \int \mathrm{d}^n \vec{q} \int \mathrm{d}^n \vec{v} \, e^{-\beta H_n(\vec{q}, \vec{v})} \;,
\end{equation}
where $\vec{q} = \begin{bmatrix} q_0 & \dots & q_{n-1} \end{bmatrix}^\mathrm{T}$ is the vector of bead positions and $\vec{v}$ the corresponding vector of velocities.
The ring-polymer Hamiltonian is given by
\begin{equation} \label{eq:rpham}
    H_n(\vec{q}, \vec{v}) = H_n^0(\vec{q}, \vec{v}) + V^{\textrm{ext}}_n(\vec{q}) \;,
\end{equation}
which includes contributions from the physical potential
\begin{equation}
    V^{\textrm{ext}}_n(\vec{q}) = \frac{1}{n} \sum_{j=0}^{n-1} V(q_j)
\end{equation}
and the free ring-polymer Hamiltonian
\begin{equation}
    H_n^0(\vec{q}, \vec{v}) = \frac{m_n}{2} \sum_{j=0}^{n-1} \left[ v_{j}^2 + \omega_n^2 (q_{j+1} - q_{j})^2 \right] \;,
\end{equation}
where $m_n = m/n$, $\omega_n = n/(\hbar\beta)$ and $q_{n} = q_{0}$.

T-RPMD evolves the phase $\begin{bmatrix} \vec{q}^\mathrm{T} & \vec{v}^\mathrm{T} \end{bmatrix}^\mathrm{T}$ of the ring polymer as per
\begin{equation} \label{eq:pile}
\begin{aligned}
    \dot{\vec{q}}(t) = \vec{v}(t) \;; ~
    &\dot{\vec{v}}(t) = - \m{\Omega}^2 \vec{q}(t) + m_n^{-1} \vec{F}( \vec{q}(t) ) \\
    & \quad - \m{\Gamma} \vec{v}(t) + \sqrt{2 \beta^{-1} m_n^{-1}} \m{\Gamma}^{1/2} \dot{\vec{W}}(t) \;,
\end{aligned}
\end{equation}
which is a coupling of the Hamitonian dynamics of $H_n(\vec{q}, \vec{v})$ with a Ornstein--Uhlenbeck thermostat.
In Eq.~\ref{eq:pile} we introduced $\vec{F}(\vec{q}) = -\nabla V^{\textrm{ext}}_n(\vec{q})$, an $n$-dimensional standard Brownian motion $\vec{W}(t)$ and the $n \times n$ matrices
\begin{equation} \label{eq:gammaomega}
\begin{aligned}
   \m{\Omega} \ &= \ \m{U} \operatorname{diag}\left(0, \omega_{1,n}, \ldots, \omega_{n-1,n}\right) \m{U}^{\mathrm{T}} \;\text{and} \\
   \m{\Gamma} \ &= \ \m{U} \operatorname{diag}\left(0, \gamma_{1,n}, \ldots, \gamma_{n-1,n}\right) \m{U}^{\mathrm{T}} \;,
\end{aligned}
\end{equation}
where $\gamma_{j,n} \ge 0$ is the $j$th friction coefficient, $\m{U}$ the $n\times n$ real discrete Fourier transform matrix, and the ring-polymer frequencies are given by
\begin{equation} \label{eq:eigenvalues}
    \omega_{j,n} =
    \begin{cases}
    2 \omega_n \sin \left( \frac{\pi j}{2 n} \right) & \text{if $j$ is even} \;, \\
    2 \omega_n \sin \left( \frac{\pi (j+1)}{2 n} \right)  & \text{else} \;.
    \end{cases}
\end{equation}
Observe that the zero-frequency (i.e., centroid) ring-polymer mode is uncoupled from the thermostat, and the coefficients $\{\gamma_{j,n}\}_{j=1}^{n-1}$ in Eq.~\ref{eq:gammaomega} constitute the friction schedule applied to the non-centroid modes.

Numerical integrators for Eq.~\ref{eq:pile} typically employ symmetric propagator splittings of the form\cite{Bussi2007,Leimkuhler2013,BouRabee2014}
\begin{equation} \label{eq:splitwithnoise}
\begin{aligned}
 e^{\Delta t \mathcal{L}_n}
 &\approx
 e^{a \frac{\Delta t}{2}\mathcal{O}_n}
 e^{\frac{\Delta t}{2}\mathcal{B}_n}
 e^{\frac{\Delta t}{2} \mathcal{A}_n}
 e^{(1-a) \Delta t\mathcal{O}_n} \\
 &\times
 e^{\frac{\Delta t}{2} \mathcal{A}_n}
 e^{\frac{\Delta t}{2}\mathcal{B}_n}
 e^{a \frac{ \Delta t}{2}\mathcal{O}_n}
 \quad \text{with $a \in \{0, 1\}$,}
\end{aligned}
\end{equation} 
where the operator $\mathcal{L}_n = \mathcal{A}_n + \mathcal{B}_n + \mathcal{O}_n$ includes contributions from the $n$-bead free ring-polymer motion ($\mathcal{A}_n$), the external potential ($\mathcal{B}_n$) and the thermostat ($\mathcal{O}_n$), and $\Delta t$ is a sufficiently small time-step.
Note that the standard microcanonical RPMD integrator is recovered in the limit of zero coupling to the thermostat,\cite{Ceriotti2010} and that Eq.~\ref{eq:splitwithnoise} yields the OBABO scheme of Bussi and Parrinello~\cite{Bussi2007} if $a = 1$ and the BAOAB scheme of Leimkuhler~\cite{Leimkuhler2013} if $a = 0$.

Standard implementations of the T-RPMD splittings in Eq.~\ref{eq:splitwithnoise} use the exact free ring-polymer propagator $e^{\frac{\Delta t}{2} \mathcal{A}_n}$ to evolve the uncoupled ring-polymer modes; however, recent work by us\cite{Korol2019} showed that such implementations exhibit poor ergodicity if large numbers $n$ of ring-polymer beads are employed in conjunction with large time-steps $\Delta t$, and suggested replacing the exact ring-polymer propagator with its Cayley approximation\cite{BouRabee2017} for improved performance.
Follow-up work\cite{Korol2020} introduced a Cayley-modified BAOAB scheme, denoted BCOCB, and presented numerical evidence that cemented the scheme as an improvement over standard BAOAB due to its superior equilibrium accuracy and time-step stability.

Generalizing beyond the Cayley modification, the current work studies a family of modified BAOAB schemes that contains BCOCB and introduces others with similar theoretical guarantees.
Specifically, the BAOAB modifications are obtained by replacing the exact free ring-polymer update in Eq.~\ref{eq:splitwithnoise} with approximations that endow the properties listed below.
\begin{enumerate}
    \item[(P1)] \textit{Strong stability.}
    For a free ring polymer (i.e., for $V(q) = \textrm{const.}$), the integrator with $\gamma_{j,n} = 0$ is both strongly stable and second-order accurate in $\Delta t$.
    \item[(P2)] \textit{Free ring-polymer ergodicity.}
    For a free ring polymer, the integrator with $\gamma_{j,n} > 0$ is ergodic with respect to the distribution with density proportional to $e^{-\beta H_n^0(\vec{q}, \vec{v})}$.  
    \item[(P3)] \textit{Dimension-free stability.}
    For a harmonically confined ring polymer (i.e., for $V(q) = (\Lambda/2) \, q^2$), the integrator with $\gamma_{j,n} = 0$ is stable for any $n$ if $\Delta t$ leads to stable integration for $n = 1$.
    \item[(P4)] \textit{Dimension-free ergodicity.}
    For a harmonically confined ring polymer, the integrator with $\gamma_{j,n} > 0$ and stable $\Delta t$ is ergodic with respect to its stationary distribution for any $n$.
    \item[(P5)] \textit{Dimension-free equilibrium accuracy.}
    For a harmonically confined ring polymer, the integrator leaves invariant an accurate approximation of the distribution with density proportional to $e^{ -\frac{\beta m_n}{2} \vec{q}^\mathrm{T} \left( \frac{\Lambda}{m} + \m{\Omega}^2 \right) \vec{q} }$, with bounded error for any $n$.
\end{enumerate}

To obtain integrators satisfying properties~(P1)-(P5), we introduce a function $\theta$ that defines the free ring-polymer update and then construct $\theta$ accordingly.
To this end, let
\begin{equation} \label{eq:gen_cay}
    \m{\mathcal{S}}_{j,n}^{1/2} \ = \ \m{\mathcal{Q}}_{j,n}
	\begin{bmatrix}
	e^{i \theta( \omega_{j,n} \Delta t )/2} & 0 \\
	0 & e^{-i \theta(\omega_{j,n} \Delta t)/2}
	\end{bmatrix}
	\m{\mathcal{Q}}_{j,n}^{-1} \;,
\end{equation}
where $\m{\mathcal{Q}}_{j,n} = \begin{bmatrix} 1 & 1 \\ i \omega_{j,n} & - i \omega_{j,n} \end{bmatrix}$ and essential properties of $\theta$ are determined in the sequel.
We focus on T-RPMD schemes derived from the BAOAB splitting (i.e., $a = 0$ in Eq.~\ref{eq:splitwithnoise}) with the exact free ring-polymer update replaced by $\m{\mathcal{S}}_{j,n}^{1/2}$.
For such schemes, an integration time-step is comprised by the following sequence of sub-steps:
\begin{enumerate}
\item[(B)] Update velocities for half a step: $\vec{v} \leftarrow \vec{v} + \frac{\Delta t}{2} \frac{\vec{F}}{m_n}$.
\item[] Convert bead Cartesian coordinates to normal modes using 
 \begin{equation} \label{eq:nm_transform}
     \vec{\varrho} = \m{U}^{\mathrm{T}} \vec{q} \qquad \text{and} \qquad 
     \vec{\varphi} = \m{U}^{\mathrm{T}} \vec{v} \;.
 \end{equation}
\item[(A)] Evolve
the free ring polymer in normal-mode coordinates for half a step:
 \begin{equation*}
     \begin{bmatrix} \varrho_j \\ \varphi_j \end{bmatrix}
     \leftarrow \m{\mathcal{S}}_{j,n}^{1/2} \begin{bmatrix} \varrho_j \\ \varphi_j \end{bmatrix} \quad \text{for $0 \le j \le n-1$.}
 \end{equation*}
\item[(O)] Perform an Ornstein--Uhlenbeck velocity update for a full time-step: 
 \begin{equation*}
     \varphi_j \leftarrow e^{- \gamma_{j,n} \Delta t} \varphi_j + \sqrt{\frac{1 - e^{-2 \gamma_{j,n} \Delta t}}{\beta m_n}} \xi_j \;,
 \end{equation*}
 where $\xi_j$ are independent standard normal random variables and $0 \le j \le n-1$.
\item[(A)] Evolve
the free ring polymer in normal-mode coordinates for half a step:
 \begin{equation*}
	 \begin{bmatrix} \varrho_j \\ \varphi_j \end{bmatrix}
	 \leftarrow \m{\mathcal{S}}_{j,n}^{1/2} \begin{bmatrix} \varrho_j \\ \varphi_j \end{bmatrix} \quad \text{for $0 \le j \le n-1$.}
 \end{equation*}
\item[] Convert back to bead Cartesian coordinates using the inverse of $\m{U}$, which is just its transpose since $\m{U}$ is orthogonal.
\item[(B)] Update velocities for half a step: $\vec{v} \leftarrow \vec{v} + \frac{\Delta t}{2} \frac{\vec{F}}{m_n}$.
\end{enumerate}

In the remainder of this section, we identify conditions on the choice of $\theta$ that imply properties~(P1)-(P5) for the corresponding T-RPMD integrator.
Despite our focus on BAOAB-like splittings, we describe how the conditions on $\theta$ can be adjusted to construct integrators derived from the OBABO splitting (i.e., $a = 1$ in Eq.~\ref{eq:splitwithnoise}) that satisfy properties~(P1)-(P5).

\subsection{Strong stability of RPMD with a constant external potential}
\label{rpmd_free_particle}

In this section, sufficient conditions on $\theta$ are identified to satisfy property~(P1) in Section~\ref{trpmd}.
Let $V(q) = \textrm{const.}$ and $\gamma_{j,n} = 0$ for $1 \le j \le n-1$, corresponding to the free ring polymer.
The $j$th normal mode $\begin{bmatrix} \varrho_j & \varphi_j \end{bmatrix}^\mathrm{T}$ satisfies
\begin{equation} \label{eq:free}
    \begin{bmatrix} \dot \varrho_j \\ \dot \varphi_j \end{bmatrix} = \m{A}_{j,n} \begin{bmatrix} \varrho_j \\ \varphi_j \end{bmatrix} \; \quad \text{where} \quad \m{A}_{j,n} = \begin{bmatrix} 0 & 1 \\ -\omega_{j,n}^2 & 0 \end{bmatrix} \;. 
\end{equation}
In this case, the algorithm from Section~\ref{trpmd} reduces to a full step of $\m{\mathcal{S}}_{j,n} \approx \Exp(\Delta t \m{A}_{j,n})$, i.e.,
\begin{equation} \label{eq:G_1D}
    \begin{bmatrix} \varrho_j \\ \varphi_j \end{bmatrix}
    \leftarrow \m{\mathcal{S}}_{j,n} \begin{bmatrix} \varrho_j \\ \varphi_j \end{bmatrix} \quad \text{for $0 \le j \le n-1$,}
\end{equation}
where $\m{\mathcal{S}}_{j,n} = \m{\mathcal{S}}_{j,n}^{1/2} \m{\mathcal{S}}_{j,n}^{1/2}$ follows from Eq.~\ref{eq:gen_cay} and the function $\theta$ is such that property~(P1) holds.

We proceed to identify sufficient conditions on $\theta$ such that the corresponding free ring-polymer update satisfies property~(P1).
First note that for any function $\theta$ such that $\theta(-x) = -\theta(x)$ for $x > 0$, the structure of $\m{\mathcal{S}}_{j,n}^{1/2}$ guarantees that the corresponding free ring-polymer update is reversible, symplectic, and preserves the free ring-polymer Hamiltonian $H_n^0(\vec{q}, \vec{v})$.
Now, observe that $\m{\mathcal{S}}_{j,n}$ is exact if $\theta(x) = x$; therefore, second-order accuracy requires that $\theta$ approximates the identity near the origin, i.e.,
\begin{equation} \label{eq:C1}
    \theta(0) = 0, \, \theta'(0) = 1, \,  \text{and}~ \, \theta''(0) = 0.
    \tag{C1}
\end{equation}
Moreover, strong stability follows if the eigenvalues $e^{\pm i \theta(\omega_{j,n} \Delta t)}$ of $\m{\mathcal{S}}_{j,n}$ are distinct;\cite{Korol2019} to this end we require that
\begin{equation} \label{eq:C2}
    0 < \theta(x) < \pi \quad \text{for $x > 0$.}
    \tag{C2}
\end{equation}
Jointly, conditions~\eqref{eq:C1} and~\eqref{eq:C2} guarantee that the update in Eq.~\ref{eq:G_1D} satisfies property~(P1).
There are many different choices of $\theta$ that obey these conditions, e.g., $\theta(x) = \arctan(x)$, $\arccos(\sech(x))$,\footnote{%
The function $\theta(x) = \arccos(\sech(x))$ is not differentiable at the origin and hence, strictly speaking, does not satisfy condition~\eqref{eq:C1}.
Moreover, the function has even symmetry and hence fails to yield a reversible free ring-polymer update.
These formal shortcomings can be fixed by multiplying the function by $\mathrm{sign}(x)$, which we implicitly do for this and other functions $\theta$ with similar features.
}
and $2\arctan(x/2)$.
The latter choice leads to the Cayley approximation of the free ring-polymer update, as can be verified by substitution in Eq.~\ref{eq:gen_cay} and comparison of the resulting $\m{\mathcal{S}}_{j,n}^{1/2}$ with Eq.~17 in Ref.~\onlinecite{Korol2020}.
Figure~\ref{fig:example_theta} compares the eigenvalues of $\m{\mathcal{S}}_{j,n}$ with $\theta(x) = x$ and several choices of $\theta$ that meet conditions~\eqref{eq:C1} and~\eqref{eq:C2}.

\begin{figure}
    \begin{subfigure}[b]{0.45\columnwidth}
        \includegraphics[width=1.0\linewidth]{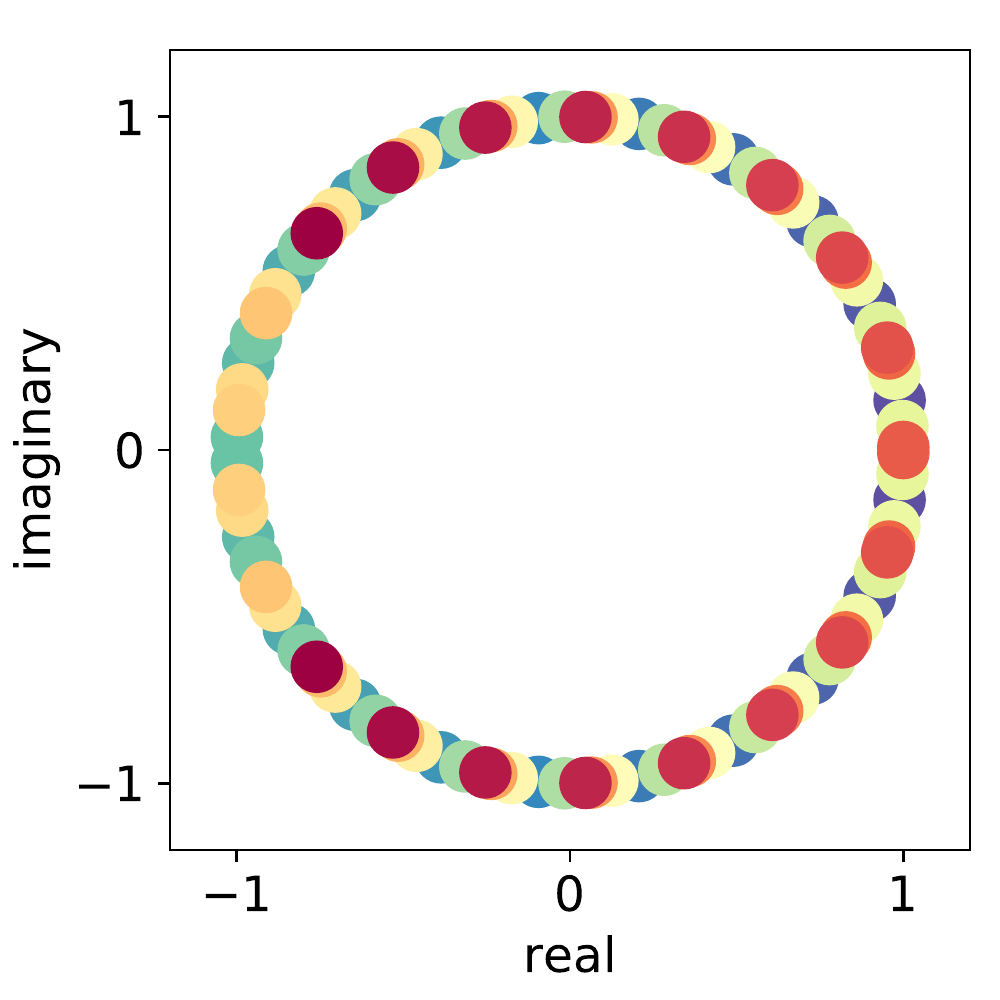}
        \caption{$\theta(x) = x$} 
    \end{subfigure}
    \begin{subfigure}[b]{0.45\columnwidth}
        \includegraphics[width=1.0\linewidth]{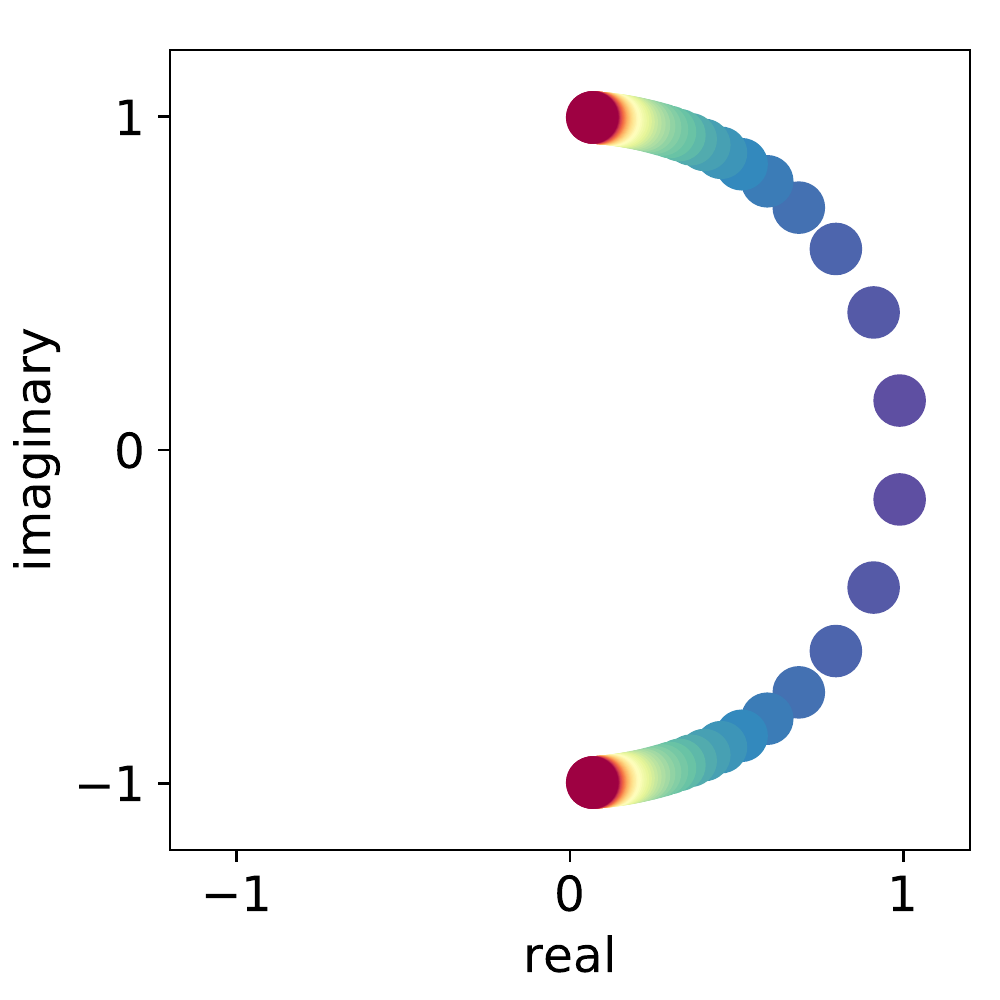}
        \caption{$\theta(x) = \arctan(x)$}
    \end{subfigure}
    \begin{subfigure}[b]{0.45\columnwidth}
        \includegraphics[width=1.0\linewidth]{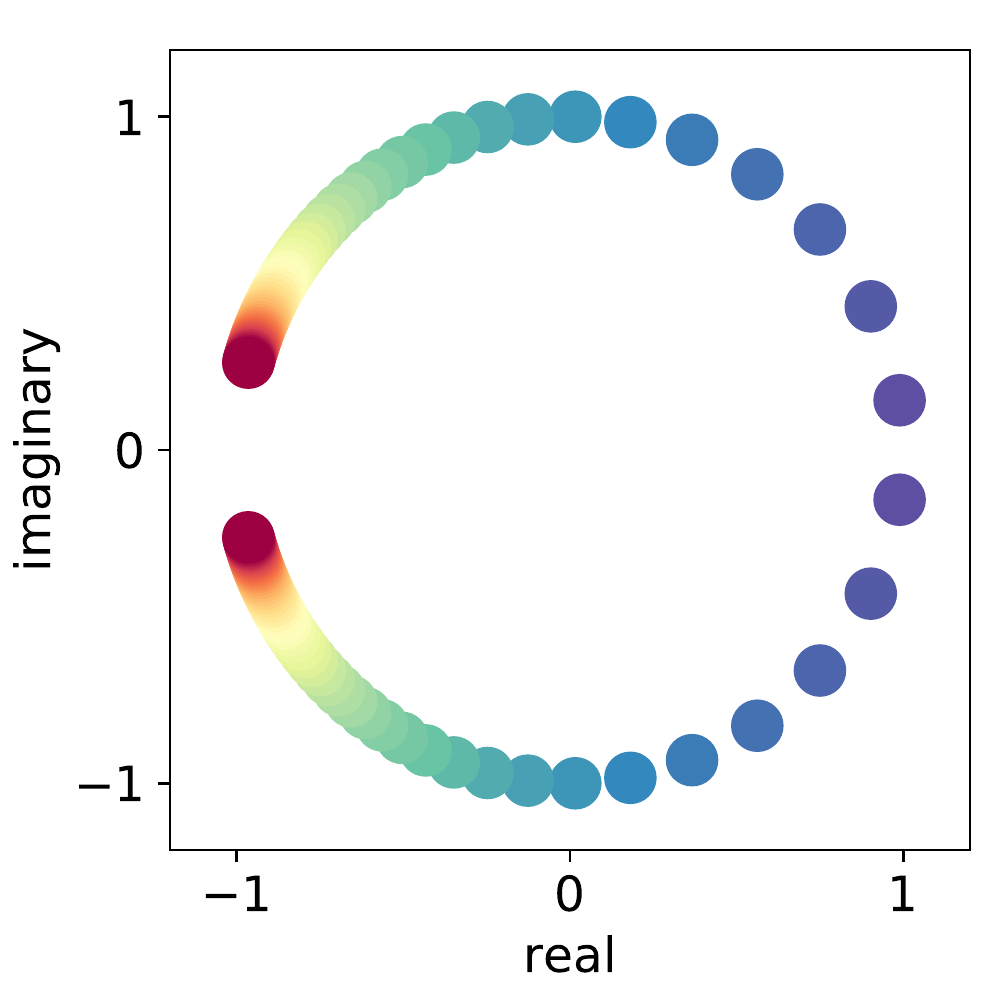}
        \caption{$\theta(x) = 2\arctan(x/2)$}
    \end{subfigure}
    \begin{subfigure}[b]{0.45\columnwidth}
        \includegraphics[width=1.0\linewidth]{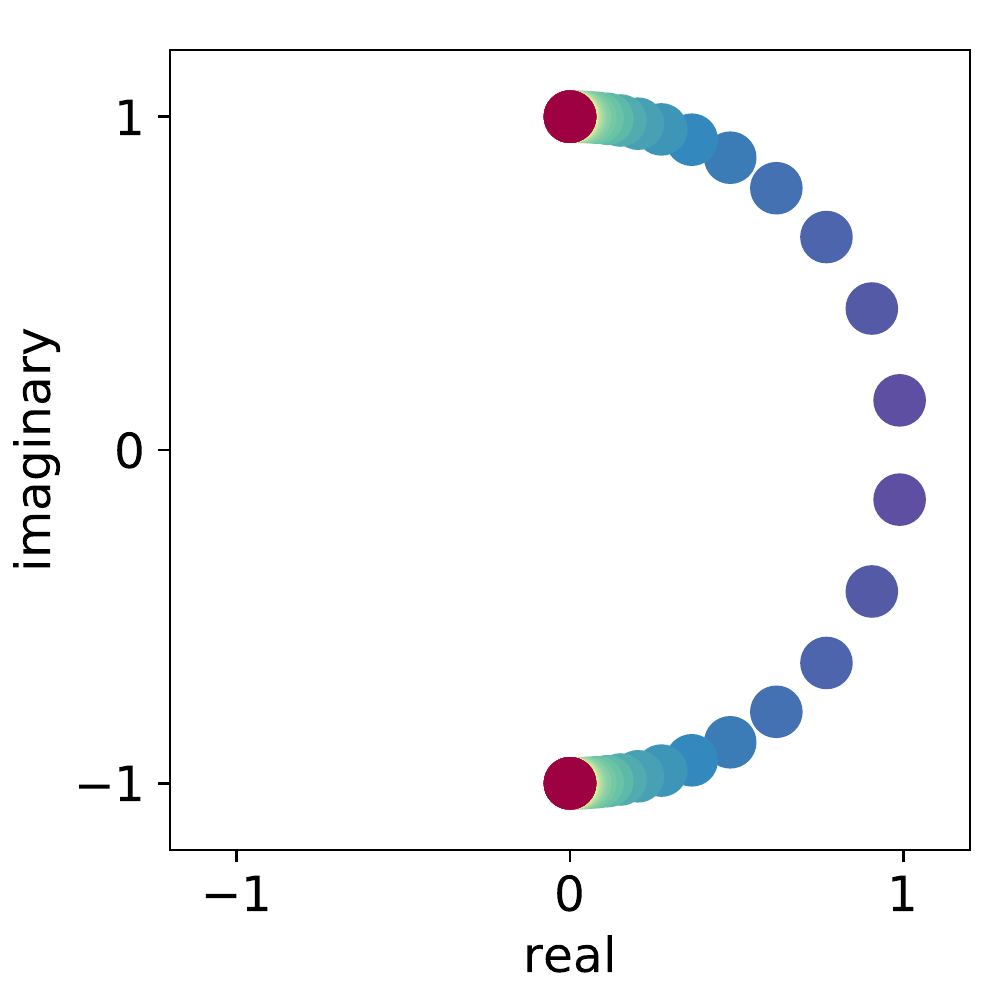}
        \caption{$\theta(x) = \arccos(\sech(x))$}
    \end{subfigure}
    \caption{\small
    Eigenvalues of $\m{\mathcal{S}}_{j,n}$ for 50 different time step sizes between 0.05 and 5.0 (evenly spaced) and fixed Matsubara frequency $\omega = 3$.
    The colors go from blue (smallest time step) through green and yellow to red (largest time step).
    In panel~(a), the eigenvalues rotate around the unit circle several times, which indicates that the corresponding $\m{\mathcal{S}}_{j,n}$ is not always strongly stable.
    In panels~(b),~(c) and~(d), the eigenvalues are distinct and on the unit circle; thus the corresponding $\m{\mathcal{S}}_{j,n}$ is strongly stable.
    } \label{fig:example_theta}
\end{figure}

\subsection{Ergodicity of T-RPMD with a constant external potential}
\label{trpmd_free_particle}

In this section, it is shown that condition~\eqref{eq:C2} implies property~(P2) in Section~\ref{trpmd}.
Let $V(q) = \textrm{const.}$ and $\gamma_{j,n} > 0$ for $1 \le j \le n-1$, corresponding to the free ring polymer with a Ornstein--Uhlenbeck thermostat.  In this case, the $j$th normal mode satisfies
\begin{equation} \label{eq:langevin_equations_1D}
    \begin{bmatrix} \dot \varrho_j \\ \dot \varphi_j \end{bmatrix} = ( \m{A}_{j,n} + \m{C}_{j,n} ) \begin{bmatrix} \varrho_j \\ \varphi_j \end{bmatrix}
    + \begin{bmatrix} 0 \\ \sqrt{ \frac{ 2 \gamma_{j,n} }{ \beta m_n } } \dot W_j \end{bmatrix} \;,
\end{equation}
where $\m{C}_{j,n} = \begin{bmatrix} 0 & 0 \\ 0 & -\gamma_{j,n} \end{bmatrix}$ and $\dot W_j$ is a scalar white-noise.
The solution $\begin{bmatrix} \varrho_j(t) & \varphi_j(t) \end{bmatrix}^\mathrm{T}$ of Eq.~\ref{eq:langevin_equations_1D} is an ergodic Markov process, and in the limit as $t \to \infty$, its distribution converges to the centered bivariate normal with covariance
\begin{equation} \label{eq:exact_IM_1D}
    \m{\Sigma}_{j,n} \ = \ \frac{1}{\beta m_n} \begin{bmatrix} s_{j,n}^2 & 0 \\ 0 & 1 \end{bmatrix} \;\; \text{where} \;\; s_{j,n}^2 \ = \ \frac{1}{\omega_{j,n}^{2}} \;.
\end{equation}
This distribution corresponds to the $j$th marginal of the free ring-polymer equilibrium distribution with density proportional to $e^{- \beta H_n^0(\vec{q}, \vec{v}) }$.

The choice of $\gamma_{j,n} > 0$ in Eq.~\ref{eq:langevin_equations_1D} determines the rate at which the associated Markov process converges to its stationary distribution if initialized away from it.
When $\gamma_{j,n} < 2 \omega_{j,n}$, the process is dominated by the deterministic Hamiltonian dynamics and is characterized as \textit{underdamped}; on the other hand, when $\gamma_{j,n} > 2 \omega_{j,n}$, the process is \textit{overdamped}; and at the critical value $\gamma_{j,n} = 2 \omega_{j,n}$ the process is characterized as \textit{critically damped} and converges to equilibrium fastest.\cite{Metafune2002,Pavliotis2014}
This analytical result motivates the so-called PILE friction schedule.\cite{Ceriotti2010,Rossi2014}
We specialize to this schedule in the remainder of the section and set $\gamma_{j,n} = 2 \omega_{j,n}$ for $1 \le j \le n-1$.

The BAOAB-like update in Section~\ref{trpmd} applied to Eq.~\ref{eq:langevin_equations_1D} can be written compactly as
\begin{equation} \label{eq:BGOGB_1D}
\begin{split}
    \begin{bmatrix} \varrho_j \\ \varphi_j \end{bmatrix}  \ \leftarrow \ \m{\mathcal{M}}_{j,n} \begin{bmatrix} \varrho_j \\ \varphi_j \end{bmatrix} + \m{\mathcal{R}}_{j,n}^{1/2} \begin{bmatrix} \xi_j \\ \eta_j \end{bmatrix}
    \\[0.5em] \text{for $0 \le j \le n-1$,}
\end{split}
\end{equation}
where $\xi_j$ and $\eta_j$ are independent standard normal random variables and we have introduced the $2 \times 2$ matrices
\begin{align*}
    \m{\mathcal{M}}_{j,n} \ &= \ \m{\mathcal{S}}_{j,n}^{1/2} \m{\mathcal{O}}_{j,n} \m{\mathcal{S}}_{j,n}^{1/2} \;, \quad \m{\mathcal{O}}_{j,n} \ = \ \begin{bmatrix} 
    1 & 0 \\
    0 & e^{-2 \omega_{j,n} \Delta t}
    \end{bmatrix} \; \text{and}
    \\
    \m{\mathcal{R}}_{j,n} \ &= \ \frac{1 - e^{-4 \omega_{j,n} \Delta t}}{\beta m_n} \m{\mathcal{S}}_{j,n}^{1/2} \begin{bmatrix}
    0 & 0 \\
    0 & 1
    \end{bmatrix} ( \m{\mathcal{S}}_{j,n}^{1/2} )^{\mathrm{T}} \;.
\end{align*}  
Since $\m{\mathcal{S}}_{j,n}^{1/2}$ and the Ornstein--Uhlenbeck update are individually preservative irrespective of the chosen $\theta$, Eq.~\ref{eq:BGOGB_1D} exactly preserves the free ring-polymer equilibrium distribution for any choice of $\theta$ that satisfies \eqref{eq:C1} and~\eqref{eq:C2}.

The ergodicity of the integrator specified by Eq.~\ref{eq:BGOGB_1D} depends entirely on the asymptotic stability of $\m{\mathcal{M}}_{j,n}$, i.e., whether or not $\| \m{\mathcal{M}}_{j,n}^k \| \to 0$ as $k \to \infty$ where $\| \cdot \|$ is a matrix norm.
The matrix $\m{\mathcal{M}}_{j,n}$ is asymptotically stable if its \textit{spectral radius} (i.e., the modulus of its largest eigenvalue) is smaller than unity,\cite{Arnol2013} which depends on
\begin{align*}
    \det(\m{\mathcal{M}}_{j,n}) &= e^{-2 \omega_{j,n} \Delta t} \quad \text{and} \\
    \tr(\m{\mathcal{M}}_{j,n}) &= \cos(\theta(\omega_{j,n} \Delta t)) (1 + e^{-2 \omega_{j,n} \Delta t}) \;.
\end{align*}
In particular, the eigenvalues of $\m{\mathcal{M}}_{j,n}$ are both inside the unit circle if and only if
\begin{equation*}
    | \tr(\m{\mathcal{M}}_{j,n}) | < 1 + \det(\m{\mathcal{M}}_{j,n}) < 2 \;;
\end{equation*}
a proof of this claim is provided in Section~\ref{eigs}.
This inequality reveals that condition~\eqref{eq:C2} implies property~(P2).
Moreover, if $\tr(\m{\mathcal{M}}_{j,n})^2 - 4 \det(\m{\mathcal{M}}_{j,n}) \le 0$, then the spectral radius of $\m{\mathcal{M}}_{j,n}$ is minimal and equal to $\sqrt{\det(\m{\mathcal{M}}_{j,n})} = e^{- \omega_{j,n} \Delta t}$; this occurs when $| \cos(\theta(\omega_{j,n} \Delta t)) | \le \sech(\omega_{j, n} \Delta t)$ for all $\omega_{j,n} \Delta t$, which holds if the function $\theta$ satisfies
\begin{equation} \label{eq:thetacrit}
\begin{split}
    \arccos(\sech(x)) \le \theta(x) \le \pi - \arccos(\sech(x)) \\
    \text{for $x > 0$.}
\end{split}
\end{equation}
Any choice of $\theta$ that does not satisfy Eq.~\ref{eq:thetacrit} will be overdamped in some modes, in the sense that the corresponding $\m{\mathcal{M}}_{j,n}$ will have a spectral radius strictly larger than $e^{- \omega_{j,n} \Delta t}$.

The function $\theta(x) = \arccos(\sech(x))$ saturates the (left) inequality in Eq.~\ref{eq:thetacrit} while satisfying conditions~\eqref{eq:C1} and~\eqref{eq:C2}, and hence provides a strongly stable and critically damped integrator for the thermostatted free ring polymer.
As illustration of this, Fig.~\ref{fig:theta_specrad}a shows that $\theta(x) = \arctan(x)$ is overdamped for all modes whereas the Cayley angle $\theta(x) = 2\arctan(x/2)$ exhibits mixed damping.
In contrast, the function $\theta(x) = \arccos(\sech(x))$ preserves the critically damped behavior of its continuous counterpart under the PILE friction schedule.
Figure~\ref{fig:theta_specrad}b confirms that the spectral radius of $\m{\mathcal{M}}_{j,n}$ is minimal at $\theta(x) = \arccos(\sech(x))$ for $x > 0$; consequently, this choice of $\theta$ optimizes the convergence of the integrator to stationarity.
 
Conditions~\eqref{eq:C1} and~\eqref{eq:C2} also imply property~(P2) for the OBABO-like update associated with a compliant choice of $\theta$, because the matrices $\m{\mathcal{S}}_{j,n}^{1/2} \m{\mathcal{O}}_{j,n} \m{\mathcal{S}}_{j,n}^{1/2}$ and $\m{\mathcal{O}}_{j,n}^{1/2} \m{\mathcal{S}}_{j,n} \m{\mathcal{O}}_{j,n}^{1/2}$ have equal spectral radii.

\begin{figure}
    \centering
    \includegraphics[width=0.85\columnwidth]{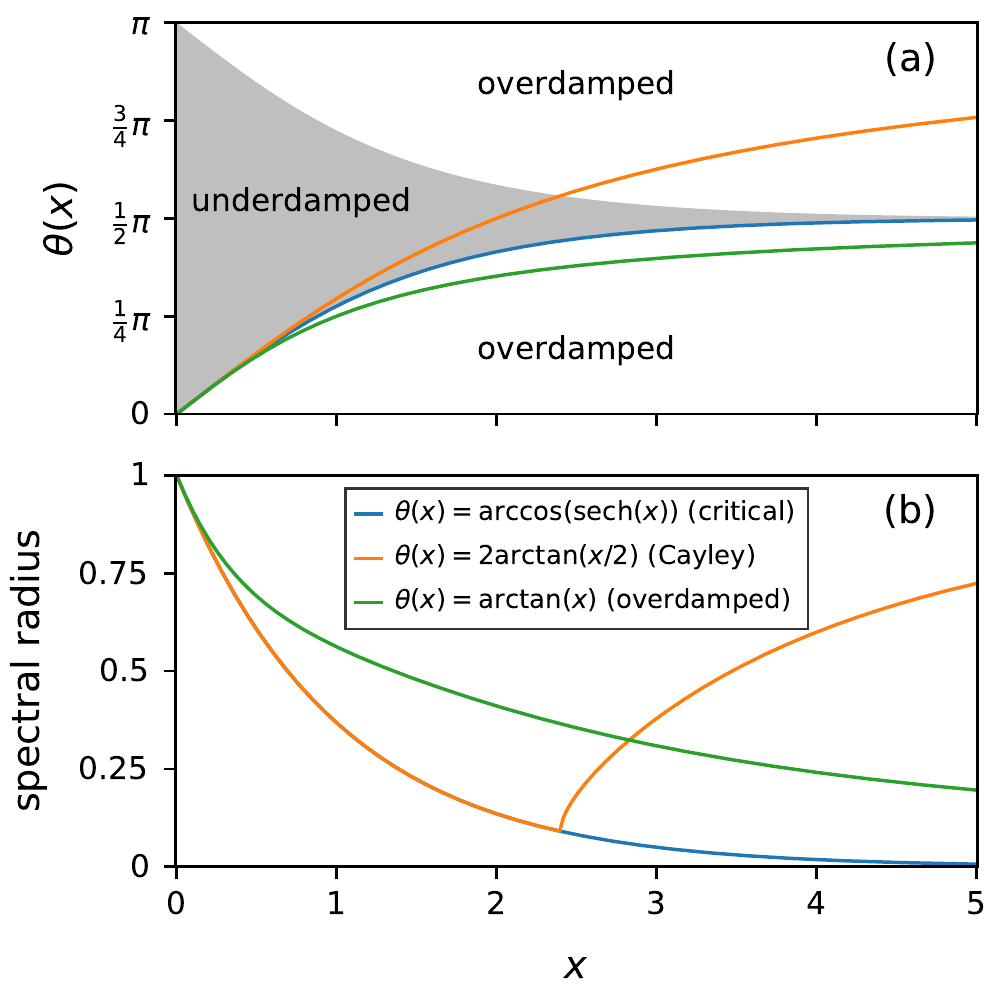}
    \caption{\small
    Spectral properties of the T-RPMD update for the free ring polymer for various choices of $\theta$.
    Panel~(a) plots the functions $\theta(x) = \arccos(\sech(x))$, $\arctan(x)$ and $\arctan(x/2)$, and regions of overdamping and underdamping with PILE friction, separated at the locus of points where $|\cos(\theta(x))\cosh(x)| = 1$.
    The gray region ($| \cos(\theta(x)) \cosh(x) | < 1$) is where the dynamics is underdamped, while in the white region ($| \cos(\theta(x)) \cosh(x) | > 1$) the dynamics is overdamped.
    The function $\theta(x) = \arctan(x)$ lies in the overdamped region for $x > 0$, whereas $\theta(x) = 2\arctan(x/2)$ is in the underdamped region for $x \lessapprox 2.4$ and in the overdamped region otherwise.
    The function $\theta(x) = \arccos(\sech(x))$, however, is critically damped for $x > 0$ and optimizes the convergence rate of the integrator.
    Panel~(b) plots the spectral radius of $\m{\mathcal{M}}_{j,n}$ corresponding to each choice of $\theta$ as a function of $x$.
    }
    \label{fig:theta_specrad}
\end{figure}

\subsection{Dimension-free stability of RPMD with a harmonic external potential}
\label{rpmd_harmonic_oscillator}

In this section, we identify a condition on $\theta$ that yields property~(P3) in Section~\ref{trpmd}.
Let $V(q) = (\Lambda/2) \, q^2$ and $\gamma_{j,n} = 0$ for $1 \le j \le n-1$, corresponding to the non-thermostatted ring polymer with a harmonic external potential.
In this case, the $j$th normal mode satisfies
\begin{equation} \label{eq:harmonic}
    \begin{bmatrix} \dot \varrho_j \\ \dot \varphi_j \end{bmatrix} = ( \m{A}_{j,n} + \m{B} ) \begin{bmatrix} \varrho_j \\ \varphi_j \end{bmatrix}
\end{equation}
where $\m{B} = \begin{bmatrix} 0 & 0 \\ -\Lambda/m & 0 \end{bmatrix}$, and conserves the Hamiltonian
\begin{equation*}
    H_{j,n}(\rho_j, \varphi_j) = \frac{m_n}{2} \big( |\varphi_j|^2 + (\omega_{j,n}^2 + \Lambda/m) |\varrho_j|^2 \big) \;.
\end{equation*}
For this system, the BAOAB-like update in Section~\ref{trpmd} reduces to
\begin{equation} \label{eq:BGB_1D}
    \begin{bmatrix} \varrho_j \\ \varphi_j \end{bmatrix}  \ \leftarrow  \ \m{\mathcal{M}}_{j,n} \begin{bmatrix} \varrho_j \\ \varphi_j \end{bmatrix}
    \quad \text{for $0 \le j \le n-1$,}
\end{equation}
where we have introduced the $2 \times 2$ matrices
\begin{equation*}
    \m{\mathcal{M}}_{j,n} = \m{\mathcal{B}}^{1/2} \m{\mathcal{S}}_{j,n} \m{\mathcal{B}}^{1/2} \;\; \text{and} \;\;
    \m{\mathcal{B}}^{1/2} = \begin{bmatrix} 
    1 & 0 \\
    -\Delta t (\Lambda/m) / 2 & 1
    \end{bmatrix} \;.
\end{equation*} 
This update may be interpreted as a symplectic perturbation of the free ring-polymer update in Eq.~\ref{eq:G_1D} due to the harmonic external potential,\cite{Korol2019} and conserves a modification of $H_{j,n}$ that depends on the choices of $\theta$ and $\Delta t$.\cite{SanzSerna1994}

The update in Eq.~\ref{eq:BGB_1D} is stable if\cite{BouRabee2018}
\begin{equation} \label{eq:tC4}
    \max_{0 \le j \le n-1} \frac{1}{2} |\tr(\m{\mathcal{M}}_{j,n})| = \max_{0 \le j \le n-1} | \mathcal{A}_{j,n} | < 1 \;,
\end{equation}
where
\begin{equation*}
    \mathcal{A}_{j,n} = \cos(\theta(\omega_{j,n} \Delta t)) - \frac{\Delta t^2 (\Lambda/m)}{2} \frac{\sin(\theta(\omega_{j,n} \Delta t))}{\omega_{j,n} \Delta t} \;.
\end{equation*}
Moreover, the $0$th (i.e., centroid) mode, like the single-bead ring polymer, evolves through the velocity Verlet algorithm, whose stability requires that $\Delta t^2 \Lambda / m < 4$.
Combining this requirement with condition~\eqref{eq:C2} yields a sufficient condition for Eq.~\ref{eq:tC4} to hold at any $n$,
\begin{equation} \label{eq:C3}
    0 < \theta(x) \le 2 \arctan(x/2) \quad \text{for $x > 0$.}
    \tag{C3}
\end{equation}
A proof of this result is provided in Section~\ref{ergodicity}.
The functions $\theta(x) = 2\arctan(x/2)$, $\arctan(x)$ and $\arccos(\sech(x))$ all satisfy condition \eqref{eq:C3}, which ensures that the corresponding RPMD integrator meets property~(P3).

\subsection{Dimension-free ergodicity and equilibrium accuracy of T-RPMD with a harmonic external potential}
\label{trpmd_harmonic_oscillator}

In this section, it is shown that condition~\eqref{eq:C3} implies property~(P4) in Section~\ref{trpmd}, and an additional condition is introduced to ensure that (P5) holds.
Let $V(q) = (\Lambda/2) \, q^2$ and $\gamma_{j,n} = 2\omega_{j,n}$ for $1 \le j \le n-1$.  In this case, the $j$th normal mode satisfies
\begin{equation} \label{eq:langevin_equations_1D_2}
    \begin{bmatrix} \dot \varrho_j \\ \dot \varphi_j \end{bmatrix} = ( \m{A}_{j,n} + \m{B} + \m{C}_{j,n} ) \begin{bmatrix} \varrho_j \\ \varphi_j \end{bmatrix}
    + \begin{bmatrix} 0 \\ \sqrt{ \frac{ 4 \omega_{j,n} }{ \beta m_n } }  \dot W_j \end{bmatrix} \;.
\end{equation}
The solution $\begin{bmatrix} \varrho_j(t) & \varphi_j(t) \end{bmatrix}^\mathrm{T}$ of Eq.~\ref{eq:langevin_equations_1D_2} is an ergodic Markov process, and its distribution as $t \to \infty$ converges to the centered bivariate normal with covariance matrix
\begin{equation} \label{eq:exact_IM_1D_2}
    \m{\Sigma}_{j,n} = \frac{1}{\beta m_n} \begin{bmatrix} s_{j,n}^2 & 0 \\ 0 & 1 \end{bmatrix} \;\; \text{where} \;\; s_{j,n}^2 = \frac{1}{\Lambda/m + \omega_{j,n}^{2}} \;;
\end{equation}
the associated position-marginal is the $j$th marginal of the ring-polymer configurational distribution with density $e^{- \frac{\beta m_n}{2} \vec{q}^\mathrm{T} \left( \frac{\Lambda}{m} + \m{\Omega}^2 \right) \vec{q} }$.

For this system, the BAOAB-like update in Section~\ref{trpmd} is of the same form as Eq.~\ref{eq:BGOGB_1D} with
\begin{equation} \label{eq:BGOGB_1D_HO}
\begin{aligned}
    \m{\mathcal{M}}_{j,n} &= \m{\mathcal{B}}^{1/2} \m{\mathcal{S}}_{j,n}^{1/2} \m{\mathcal{O}}_{j,n} \m{\mathcal{S}}_{j,n}^{1/2} \m{\mathcal{B}}^{1/2} \quad \text{and} 
    \\
    \m{\mathcal{R}}_{j,n} &= \frac{1-e^{-4 \omega_{j,n} \Delta t} }{\beta m_n} \m{\mathcal{B}}^{1/2} \m{\mathcal{S}}_{j,n}^{1/2} \begin{bmatrix} 0 & 0 \\ 0 & 1 \end{bmatrix} ( \m{\mathcal{B}}^{1/2} \m{\mathcal{S}}_{j,n}^{1/2} )^{\mathrm{T}} \;.
\end{aligned} 
\end{equation}
As in the case of a constant external potential, the ergodicity of this integrator depends on the spectral radius of $\m{\mathcal{M}}_{j,n}$.
By Theorem~\ref{thm:eigs} in Section~\ref{eigs} and the fact that
\begin{align*}
    \det(\m{\mathcal{M}}_{j,n})
    &= e^{-2 \omega_{j,n} \Delta t} \; \quad \text{and} \\
    \tr(\m{\mathcal{M}}_{j,n})
    &= \mathcal{A}_{j,n} (1 + e^{-2 \omega_{j,n} \Delta t}) \;,
\end{align*}
it follows that condition~\eqref{eq:C3} gives a simple and sufficient condition for ergodicity at any bead number $n$ and hence implies property~(P4) for the BAOAB-like update specified by Eqs.~\ref{eq:BGOGB_1D} and~\ref{eq:BGOGB_1D_HO}.
Furthermore, because the matrix $\m{\mathcal{M}}_{j,n}$ of the corresponding OBABO-like update has equal trace and determinant, condition~\eqref{eq:C3} also guarantees property~(P4) in that case.\footnote{%
Condition~\eqref{eq:C3} may be viewed as a relaxation of the sufficient condition for ergodicity given in Eq.~(18) of Ref.~\onlinecite{Korol2020}.
Indeed, condition~\eqref{eq:C3} implies ergodicity irrespective of the Ornstein--Uhlenbeck friction schedule, whereas Eq.~(18) in Ref.~\onlinecite{Korol2020} does not imply ergodicity for friction schedules that lead to overdamped dynamics.
}

If condition~\eqref{eq:C3} holds, the BAOAB-like update is ergodic with respect to a centered bivariate normal distribution whose covariance matrix $\m{\Sigma}_{j, \Delta t}$ satisfies the linear equation
\begin{equation} \label{eq:sigma_j_dt}
    \m{\Sigma}_{j, \Delta t} = \m{\mathcal{M}}_{j,n} \m{\Sigma}_{j, \Delta t} \m{\mathcal{M}}_{j,n}^{\mathrm{T}} + \m{\mathcal{R}}_{j,n} \;,
\end{equation}
for which the solution is
\begin{equation} \label{eq:numerical_IM_1D_a}
    \m{\Sigma}_{j, \Delta t} = \frac{1}{\beta m_n} \begin{bmatrix} s_{j, \Delta t}^2  & 0 \\ 0 & r_{j, \Delta t}^2 \end{bmatrix}
\end{equation}
where the variance in the position- and velocity-marginal is respectively $(\beta m_n)^{-1} s_{j, \Delta t}^2$ and $(\beta m_n)^{-1} r_{j, \Delta t}^2$ with
\begin{equation} \label{eq:numerical_IM_1D_b}
\begin{aligned}
    s_{j, \Delta t}^2 &=
    \left( \omega_{j,n}^2 + \frac{\Lambda}{m} \frac{ \omega_{j,n} \Delta t/2 }{ \tan\left(\theta(\omega_{j,n} \Delta t) /2 \right) } \right)^{-1} \;\; \text{and}
    \\
    r_{j, \Delta t}^2 &=
    1 - \frac{\Delta t^2 \Lambda}{4 m} \frac{ \tan\left(\theta(\omega_{j,n} \Delta t)/2 \right) }{ \omega_{j,n} \Delta t/2 } \;.
\end{aligned}
\end{equation}
Because the tangent function is monotonically increasing on the range of $\theta$ specified by condition~\eqref{eq:C3}, we have the
correspondence
\begin{equation} \label{eq:exact_numerical_IM_1D_comparison}
    0 < s_{j, \Delta t}^2 \le s_{j}^2 \quad\text{and}\quad
    1 - \frac{ \Delta t^2 \Lambda }{ 4 m } \le r_{j, \Delta t}^2 < 1
\end{equation}
between the exact and numerical variances of the $j$th ring-polymer mode.
Equation~\ref{eq:numerical_IM_1D_b} reveals that $\theta(x) = 2\arctan(x/2)$ is the unique function that complies with condition~\eqref{eq:C3} and saturates the inequality $s_{j, \Delta t}^2 \le s_{j}^2$ in Eq.~\ref{eq:exact_numerical_IM_1D_comparison}; consequently, the corresponding BAOAB-like scheme preserves the exact position-marginal in all modes and trivially satisfies property~(P5).
The BCOCB integrator from Ref.~\onlinecite{Korol2020} corresponds to this choice of $\theta$ and thus uniquely provides optimal equilibrium position-marginal accuracy for harmonic external potentials.

To identify other BAOAB-like schemes compliant with condition~\eqref{eq:C3} that satisfy property~(P5), we examine the overlap between the numerical stationary position-marginal distribution $\mu_{n, \Delta t}$ and the exact distribution $\mu_{n}$ where
\begin{equation*}
    \mu_{n} = \prod_{j=1}^{n-1} \mathcal{N} \! \left( 0, \tfrac{s_{j}^2}{\beta m_n} \right) \quad \text{and} \quad
    \mu_{n, \Delta t} = \prod_{j=1}^{n-1} \mathcal{N} \! \left( 0, \tfrac{s_{j, \Delta t}^2}{\beta m_n} \right) \;.
\end{equation*}
Centroid-mode marginals have been suppressed in the definitions of $\mu_{n}$ and $\mu_{n, \Delta t}$.
A BAOAB-like scheme is dimension-free if it admits an $n$-independent upper bound on the distance $d_\mathrm{TV}(\mu_{n}, \mu_{n, \Delta t})$ between $\mu_{n}$ and $\mu_{n, \Delta t}$, where $d_\mathrm{TV}$ is the total variation metric.\cite{Gibbs2002}
In particular, if we require
\begin{equation} \label{eq:C4}
    \frac{x}{1+|x|} \le \theta(x) \le 2 \arctan(x/2) \quad \text{for $x > 0$,}
    \tag{C4}
\end{equation}
then we have the dimension-free bound
\begin{equation} \label{eq:tv_error}
    d_{\mathrm{TV}}(\mu_{n}, \mu_{n, \Delta t}) < \left( \sqrt{\frac{4}{3}} \frac{\hbar \beta}{\Delta t} \right) \frac{\Delta t^2 \Lambda}{m} \;.
\end{equation}
A proof of this claim is provided in Section~\ref{tv_error_proof}.
Condition~\eqref{eq:C4} ensures that any BAOAB-like integrator with a compliant choice of $\theta$ meets property~(P5).

For OBABO-like schemes, the bound in condition~\eqref{eq:C4} must be tightened to guarantee non-zero overlap between $\mu_{n}$ and $\mu_{n, \Delta t}$ for arbitrarily large $n$.
In particular, replacing $2\arctan(x/2)$ with $\min \{2\arctan(x/2), C\}$ for some $C \in (0, \pi)$ in the upper bound of condition~\eqref{eq:C4} yields a $n$-independent bound on $d_\mathrm{TV}(\mu_{n}, \mu_{n, \Delta t})$ for all compliant OBABO-like integrators, as
can be shown through arguments similar to those in Section~\ref{tv_error_proof}.

Jointly, conditions~\eqref{eq:C1}-\eqref{eq:C4} specify a family of BAOAB-like schemes with dimension-free stability, ergodicity and equilibrium accuracy for applications with harmonic external potentials.
Numerical results in Section~\ref{numerical_results} suggest that the integrators exhibit similar properties in a more realistic setting with a strongly anharmonic external potential.

\subsection{Dimension-free convergence to equilibrium of T-RPMD with a harmonic external potential}
\label{trpmd_harmonic_oscillator_convergence}

Beyond ensuring ergodicity of the T-RPMD update in Eq.~\ref{eq:BGOGB_1D_HO}, condition~\eqref{eq:C3} leads to explicit dimension-free equilibration rates for compliant schemes.
Theorem~\ref{thm:contr} in Section~\ref{contr} proves this result in the infinite-friction limit for ring-polymer modes with arbitrarily high frequency.
In detail, the theorem shows that the configurational transition kernel associated with the T-RPMD update of the $j$th mode in Eq.~\ref{eq:BGOGB_1D_HO} is contractive in the $2$-Wasserstein metric\cite{Villani2008} and equilibrates any given initial distribution at a rate determined by the function $\theta$, the (external) potential curvature $\Lambda$, and the (stable) time-step $\Delta t$ if condition~\eqref{eq:C3} holds.
The rate in Theorem~\ref{thm:contr}, though obtained in the infinite-friction limit, holds for finite friction coefficients $\gamma_{j,n}$ leading to spectral radii $\rho(\m{\mathcal{M}}_{j,n}) \le |\mathcal{A}_{j,n}|$, where $\mathcal{A}_{j,n}$ is defined in the display after Eq.~\ref{eq:tC4} and $|\mathcal{A}_{j,n}| = \lim\nolimits_{\gamma_{j,n} \to \infty} \rho(\m{\mathcal{M}}_{j,n})$ is the spectral radius at infinite damping.

\begin{figure}
    \centering
    \includegraphics[width=0.85\columnwidth]{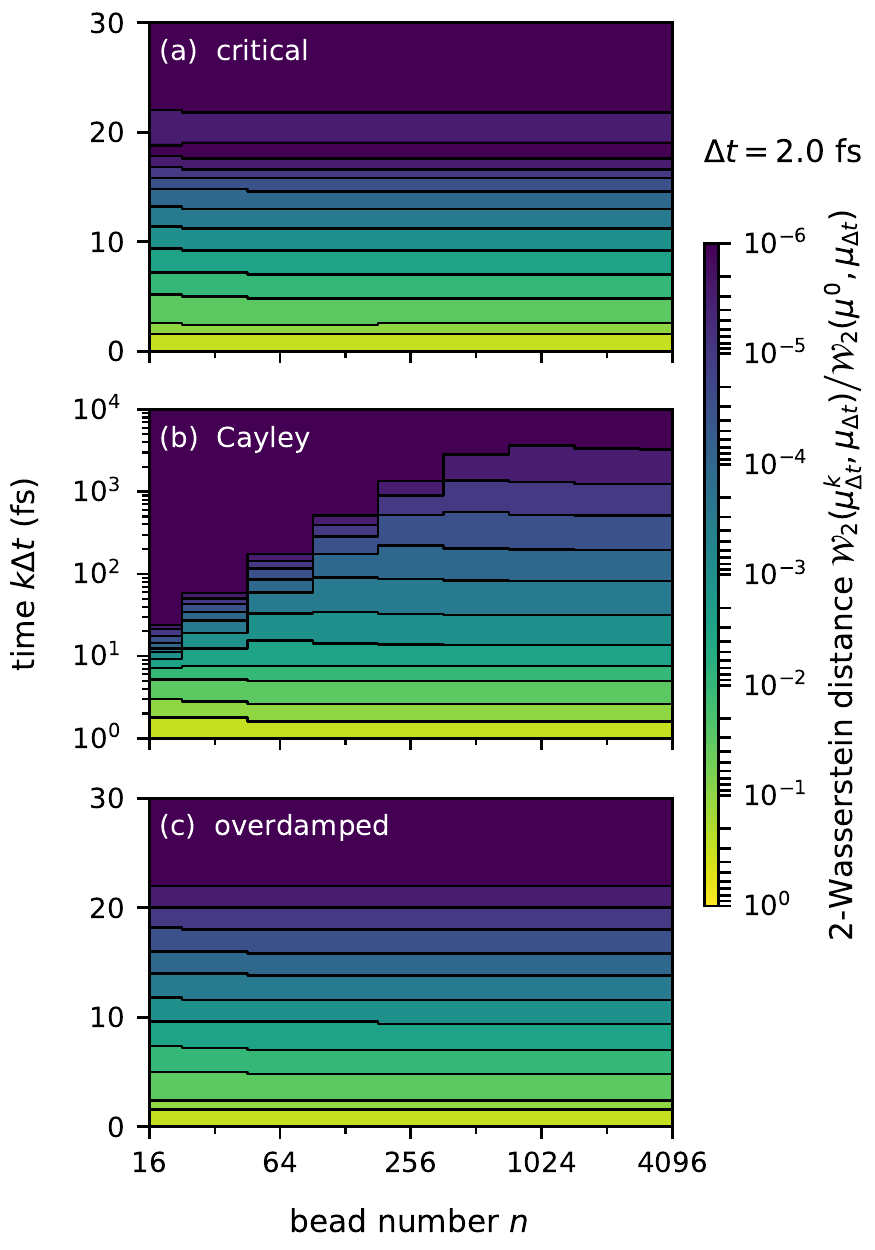}
    \caption{\small
    Dimension-free convergence to equilibrium of BAOAB-like T-RPMD schemes with a harmonic external potential.
    The physical parameters of the ring-polymer system (i.e., $\Lambda$, $m$, and $\beta$) are listed in Section~\ref{harmonic_oscillator_results}.
    Panels~(a),~(b) and~(c) plot the normalized $2$-Wasserstein distance between the configurational ring-polymer distribution at stationarity and at time $k \Delta t$, as evolved via various BAOAB-like schemes from an initial point-mass distribution.
    Regions with darker color indicate smaller $2$-Wasserstein distance to stationarity, and black lines mark iso-distance contours.
    The contours plateau at some value of $n$ for all tested schemes, which checks that they exhibit dimension-free convergence as predicted by Theorem~\ref{thm:contr}.
    }
    \label{fig:dfrate_contour}
\end{figure}

To illustrate dimension-free convergence, Fig.~\ref{fig:dfrate_contour} plots the $2$-Wasserstein distance between the stationary configurational (i.e., position-marginal) distribution $\mu_{n, \Delta t}$ and the distribution $\mu_{n, \Delta t}^k$ at the $k$th T-RPMD step evolved from a point mass at the origin using the schemes specified by $\theta(x) = \arccos(\sech(x))$ (Fig.~\ref{fig:dfrate_contour}a), $2\arctan(x/2)$ (Fig.~\ref{fig:dfrate_contour}b), and $\arctan(x)$ (Fig.~\ref{fig:dfrate_contour}c) for a range of bead numbers $n$.
These choices of $\theta$ respectively lead to \textit{overdamped}, \textit{critical}, and \textit{Cayley} evolution of the thermostatted free ring polymer under PILE damping (see Section~\ref{trpmd_free_particle}), and are identified accordingly in Fig.~\ref{fig:dfrate_contour}. 
The ring-polymer system considered in Fig.~\ref{fig:dfrate_contour} approximates the O--H stretch dynamics in liquid water at room temperature with the parameters listed in Section~\ref{harmonic_oscillator_results}.
Velocity-marginals were initialized as in the setting of Theorem~\ref{thm:contr} (see Section~\ref{contr}), and the position of the $j$th ring-polymer mode at time $k \Delta t$ follows a centered normal distribution with variance $(\beta m_n)^{-1} (s_{j, \Delta t}^{k})^2$, where
\begin{equation*}
\begin{split}
    (s_{j, \Delta t}^{k})^2 = (\m{\mathcal{M}}_{j,n}^k)_{12}^2 + \beta m_n \sum_{\ell=0}^{k-1} \big( \m{\mathcal{M}}_{j,n}^{\ell} \m{\mathcal{R}}_{j,n} (\m{\mathcal{M}}_{j,n}^{\ell})^\mathrm{T} \big)_{11} \\ \text{for $k > 0$.}
\end{split}
\end{equation*}
The $2$-Wasserstein distances in Fig.~\ref{fig:dfrate_contour} were evaluated using a well-known analytical result for multivariate normal distributions.\cite{Givens1984}

Figures~\ref{fig:dfrate_contour}a and~\ref{fig:dfrate_contour}c clearly show that the critical and overdamped schemes converge at dimension-free rates, but this is less evident from Fig.~\ref{fig:dfrate_contour}b for the Cayley scheme.
The latter scheme nonetheless displays an $n$-independent, and hence dimension-free, distance to stationarity at all times $k \Delta t > 0$, indicated by plateauing of the contour lines towards the right of Fig.~\ref{fig:dfrate_contour}b.
The ladder-like pattern that precedes this plateau illustrates a transition from geometric (i.e., fast) to sub-geometric (i.e., slow) convergence upon introducing higher-frequency modes into the ring polymer.
The transition manifests with the Cayley scheme because of its aggressive overdamping of the high-frequency modes, which is absent in the other two schemes (see Fig.~\ref{fig:theta_specrad}).

The example considered in this section illustrates that the equilibration timescale (e.g., the time until the $2$-Wasserstein distance decays below $10^{-6}$) of the Cayley scheme at large $n$ can dramatically exceed that of other BAOAB-like schemes.
Although this negative feature may render the scheme impractical for pathological applications, we find in the next section that the Cayley scheme's superior configurational sampling provides compelling justification for its preferred use in realistic settings.

\section{Numerical results}
\label{numerical_results}

The current section provides numerical comparisons of the BAOAB-like T-RPMD integrators in Section~\ref{theory}, on applications featuring harmonic (Section~\ref{harmonic_oscillator_results}) and anharmonic (Section~\ref{liquid_water_results}) external potentials.
Three representative choices of $\theta$ are considered in the numerical comparisons, namely $\theta(x) = \arctan(x)$, $\arccos(\sech(x))$, and $2\arctan(x/2)$.
These choices respectively lead to \textit{overdamped}, \textit{critical}, and \textit{Cayley} evolution of the thermostatted free ring polymer under PILE damping (Section~\ref{trpmd_free_particle}), and are identified accordingly throughout the current section.
It is borne out from the numerical comparisons that the Cayley scheme exhibits superior configurational sampling among the tested schemes in both applications.\footnote{%
The trajectory data used to produce the figures in this section (via the protocols in Section~\ref{calc_details}) is available from the corresponding author upon reasonable request.
}

\subsection{One-dimensional quantum harmonic oscillator}
\label{harmonic_oscillator_results}

In the current section, we numerically integrate Eq.~\ref{eq:pile} with the harmonic potential $V(q) = (\Lambda/2) \, q^2$ using PILE friction (i.e., $\m{\Gamma} = 2 \m{\Omega}$), $m = 0.95 \textrm{ amu}$, $\sqrt{\Lambda/m} = 3886 \textrm{ cm}^{-1}$, and $T = 298 \textrm{ K}$.
This choice of physical parameters corresponds to a harmonic approximation of the Morse contribution to the O--H bond potential in the q-TIP4P/F force field for water,\cite{Habershon2009} and sets a least upper bound for the T-RPMD stability interval at $\Delta t^\textrm{max} = 2 / \sqrt{\Lambda/m} = 2.74 \textrm{ fs}$.
The simulations reported throughout this section employ the time-step $\Delta t = 0.73 \times \Delta t^\textrm{max} = 2.00 \textrm{ fs}$.

\begin{figure}
    \centering
    \includegraphics[width=\columnwidth]{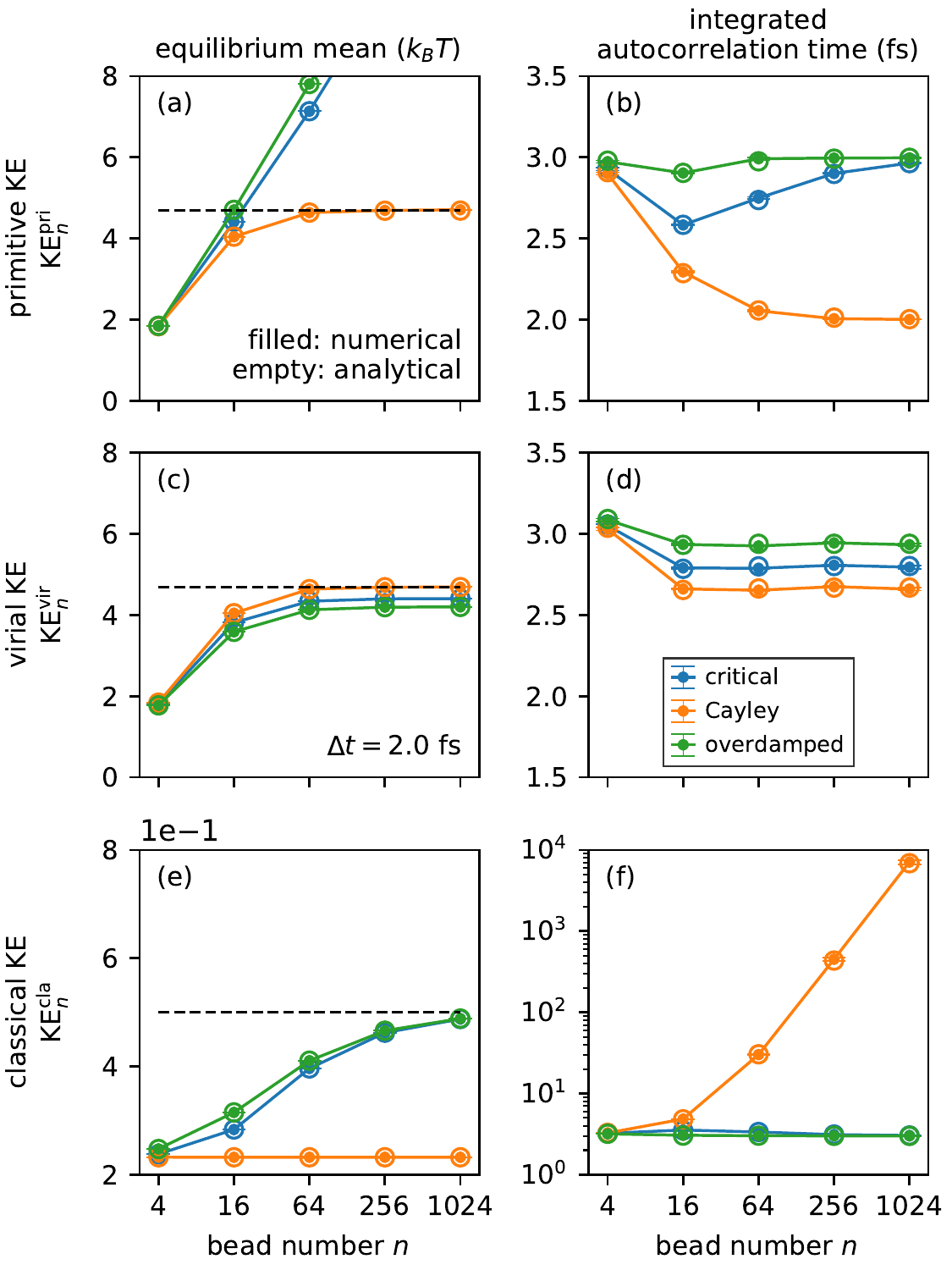}
    \caption{\small
    Performance at equilibrium of various BAOAB-like T-RPMD schemes applied to the one-dimensional quantum harmonic oscillator with physical parameters listed in Section~\ref{harmonic_oscillator_results}.
    Panels (a), (c) and (e), respectively, plot the equilibrium mean primitive kinetic energy, virial kinetic energy, and non-centroid classical kinetic energy per mode as a function of bead number $n$; the corresponding means in the exact infinite bead limit are plotted as dashed lines.
    Panels (b), (d) and (f), plot the integrated autocorrelation times (Eq.~\ref{eq:obs_iact}) of the respective observables.
    Exact (resp.\ numerically estimated) values of the plotted quantities are shown with empty (resp.\ filled) circles.
    Numerical estimates were obtained using the protocol described in Section~\ref{calc_details}.
    }
    \label{fig:harmonic_oscillator_equilibrium}
\end{figure}

Figure~\ref{fig:harmonic_oscillator_equilibrium} compares the accuracy and efficiency of various BAOAB-like T-RPMD schemes at equilibrium as a function of the bead number $n$.
For a description of the numerical simulation and statistical estimation procedures used to generate the numerical data (filled circles) in Fig.~\ref{fig:harmonic_oscillator_equilibrium}, the reader is referred to Section~\ref{calc_details}.
Figures~\ref{fig:harmonic_oscillator_equilibrium}a and~\ref{fig:harmonic_oscillator_equilibrium}c report the mean quantum kinetic energy at equilibrium as per the primitive and virial estimators,
\begin{equation} \label{eq:qke}
\begin{aligned}
    \mathsf{KE}^\mathrm{pri}_n(\vec{q}) &= \frac{n}{2 \beta} - \sum_{j=0}^{n-1} \frac{m_n \omega_n^2}{2} \, (q_{j+1} - q_j)^2 \;\; \text{and} \\
    \mathsf{KE}^\mathrm{vir}_n(\vec{q}) &= \frac{1}{2 \beta} + \frac{1}{2} \sum_{j=0}^{n-1} (q_j - \overline{q}) \, \partial_{q_j} V_n^\textrm{ext}(\vec{q}) \;,
\end{aligned}
\end{equation}
where $\overline{q} = \frac{1}{n} \sum_{j=0}^{n-1} q_j$ is the centroid position of the $n$-bead ring polymer.
For these two observables, Figs.~\ref{fig:harmonic_oscillator_equilibrium}b and~\ref{fig:harmonic_oscillator_equilibrium}d quantify the equilibrium sampling efficiency of the schemes in terms of the integrated autocorrelation time (or normalized asymptotic variance)\cite{Geyer1992,Sokal1997,Asmussen2007,Skeel2017,Fang2017}
\begin{align} \label{eq:obs_iact}
    \frac{\mathrm{aVar}(\mathsf{O}_n)}{\mathrm{Var}(\mathsf{O}_n)}
    &= \frac{ \lim_{K \to \infty} \mathrm{Var} \Big( \frac{1}{\sqrt{K}} \sum\nolimits_{k=0}^{K-1} \mathsf{O}_n(\vec{\xi}^{(k \Delta t)}) \Big) }{ \mathrm{Var} \big( \mathsf{O}_n \big) } \nonumber \\
    &= 1 + 2 \sum_{k=1}^\infty \mathrm{Cor} \big( \mathsf{O}_n(\vec{\xi}^{(0)}), \mathsf{O}_n(\vec{\xi}^{(k \Delta t)}) \big) \;,
\end{align}
where $\mathsf{O}_n$ is an $n$-bead observable, $\{ \vec{\xi}^{(k \Delta t)} \}_{k=0}^\infty = \{ (\vec{q}^{(k \Delta t)}, \vec{v}^{(k \Delta t)}) \}_{k=0}^\infty$ a stationary T-RPMD trajectory, $\mathrm{Var}(\mathsf{O}_n)$ the variance of $\mathsf{O}_n$ at equilibrium, and $\mathrm{Cor}(\mathsf{O}_n(\vec{\xi}^{(0)}), \mathsf{O}_n(\vec{\xi}^{(k \Delta t)}))$ the lag-$k \Delta t$ autocorrelation of $\mathsf{O}_n$ along the T-RPMD trajectory.
The integrated autocorrelation time of $\mathsf{O}_n$ is interpreted as the timescale over which adjacent observations along an equilibrium trajectory become statistically uncorrelated\cite{Geyer1992,Sokal1997,Asmussen2007,Skeel2017,Fang2017} and is hence a measure of the efficiency of a T-RPMD scheme at estimating the mean of $\mathsf{O}_n$ with respect to the numerically sampled equilibrium distribution.
Figures~\ref{fig:harmonic_oscillator_equilibrium}a-d show that the scheme specified by the Cayley angle (orange) outperforms others in terms of both accuracy and efficiency at estimating the equilibrium average of the quantum kinetic energy observables.

From the perspective of configurational accuracy, the optimality of the Cayley angle displayed in Figs.~\ref{fig:harmonic_oscillator_equilibrium}a and~\ref{fig:harmonic_oscillator_equilibrium}c is not surprising in light of the findings in Section~\ref{trpmd_harmonic_oscillator}.
Less expected are the results in Figs.~\ref{fig:harmonic_oscillator_equilibrium}b and~\ref{fig:harmonic_oscillator_equilibrium}d, which suggest that the Cayley angle  is also optimal from the standpoint of configurational sampling efficiency for the quantum kinetic energy observables in Eq.~\ref{eq:qke}.
Section~\ref{analytical_asymptotic_variance} supports this conjecture with an analytical result for harmonic external potentials.

Figure~\ref{fig:harmonic_oscillator_equilibrium}e plots the mean classical kinetic energy at equilibrium as computed from the non-centroid ring-polymer velocities,
\begin{equation}\label{eq:cke}
    \mathsf{KE}_n^\mathrm{cla}(\vec{v}) = \frac{m_n}{2(n-1)} \sum_{j=0}^{n-1} (v_j^2 - \overline{v}^2) \approx \frac{1}{2 \beta} \;,
\end{equation}
and Fig.~\ref{fig:harmonic_oscillator_equilibrium}f plots the corresponding integrated autocorrelation time as given by Eq.~\ref{eq:obs_iact}.
For this observable, the equilibrium accuracy and efficiency of the Cayley scheme are significantly worse than those of the others as $n$ increases.
This is a consequence of the strongly overdamped behavior of Cayley T-RPMD at high frequencies (see Fig.~\ref{fig:theta_specrad}), for which the integrator's ergodicity degrades as its spectral radius approaches unity.
Note that this shortcoming of the Cayley scheme presents no adverse implications to the equilibrium sampling of observables that exclusively depend on the ring-polymer configuration, as confirmed by Figs.~\ref{fig:harmonic_oscillator_equilibrium}a-d.

In summary, Fig.~\ref{fig:harmonic_oscillator_equilibrium} establishes that the T-RPMD scheme specified by the Cayley angle provides optimally accurate and efficient configurational sampling \textit{at} equilibrium.
To exploit this remarkable feature in practice, the scheme must manifest rapid converge \textit{to} equilibrium when initialized away from it, as is necessary in most realistic applications of T-RPMD.
Fortunately, Theorem~\ref{thm:contr} guarantees that any BAOAB-like scheme compliant with conditions~\eqref{eq:C1}-\eqref{eq:C4} features a contractive configurational transition kernel for any number of ring-polymer beads, and Fig.~\ref{fig:dfrate_contour} in Section~\ref{trpmd_harmonic_oscillator_convergence} illustrates this fact  for the quantum harmonic oscillator considered in the current section.

\begin{figure}
    \centering
    \includegraphics[width=\columnwidth]{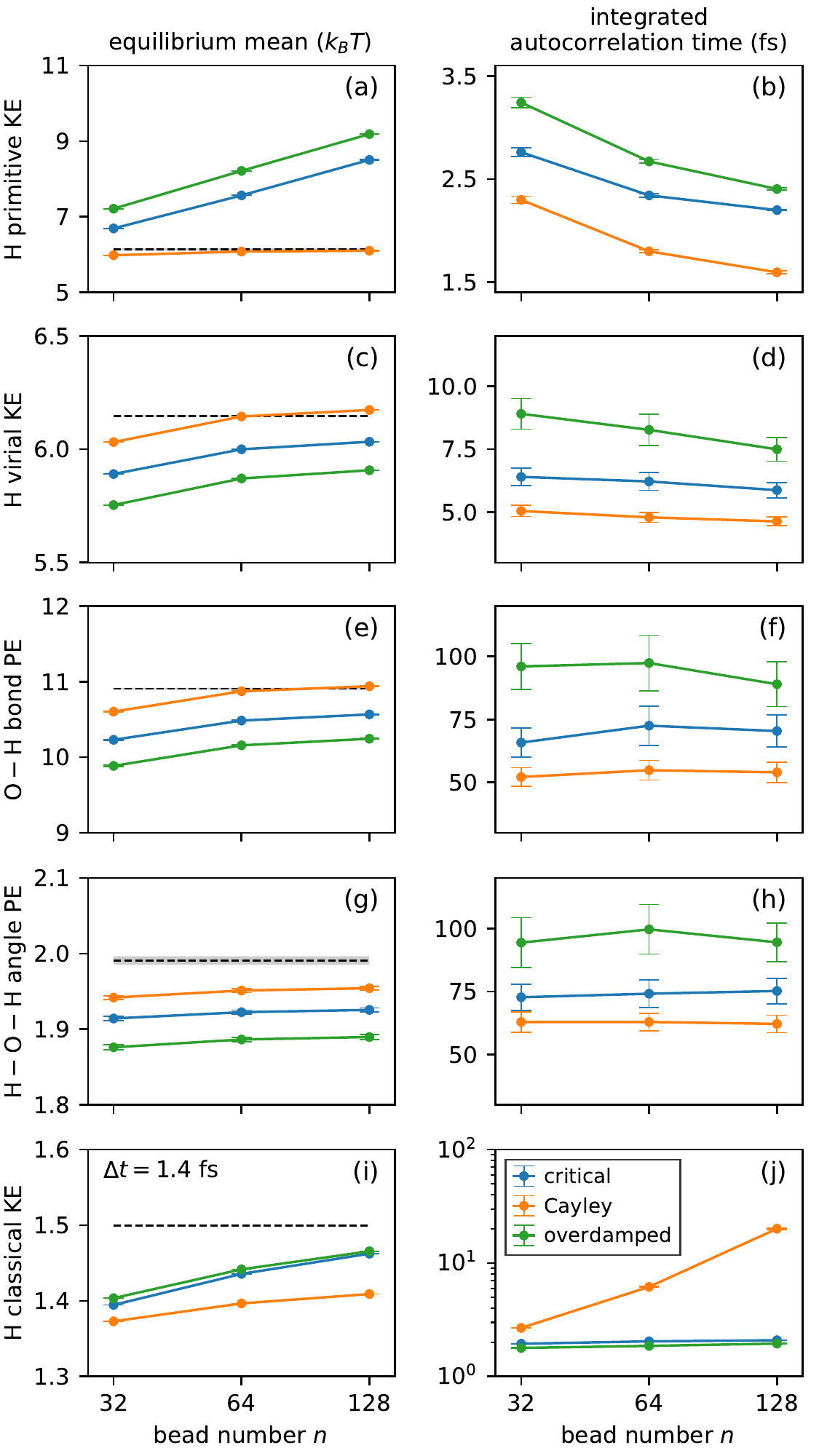}
    \caption{\small 
    Performance of various BAOAB-like T-RPMD schemes applied to q-TIP4P/F liquid water at room temperature.
    As a function of the bead number $n$ and for a $1.4$-fs time-step, panels~(a) and~(c) plot the equilibrium kinetic energy per $\mathrm{H}$ atom as per the primitive and virial estimators (Eq.~\ref{eq:qke}), and panels~(b) and~(d) plot the corresponding integrated autocorrelation times.
    Likewise, panels~(e) and~(g) plot the equilibrium potential energy per $\mathrm{H}_2\mathrm{O}$ molecule due to the $\mathrm{O}\!-\!\mathrm{H}$-stretch and $\mathrm{H}\!-\!\mathrm{O}\!-\!\mathrm{H}$-bend contributions, as defined in the q-TIP4P/F force field,\cite{Habershon2009} and the corresponding autocorrelation times are plotted by panels~(f) and~(h).
    Finally, panel~(i) plots the classical kinetic energy per $\mathrm{H}$ atom computed from the non-centroid velocity estimator (Eq.~\ref{eq:cke}), and panel~(j) plots the corresponding autocorrelation time.
    The numerical estimates and reference results (dashed lines) were obtained using the protocols described in Section~\ref{calc_details}.
    }
    \label{fig:liquid_water_equilibrium_all}
\end{figure}

\subsection{Room-temperature liquid water}
\label{liquid_water_results}

While theoretical analysis and numerical tests of BAOAB-like T-RPMD schemes in previous sections have focused on harmonic external potentials, the current section demonstrates that the resulting insights carry over to a realistic, strongly anharmonic model of room-temperature liquid water.
Our test system is a periodic box containing $32$ water molecules at a temperature of $298 \textrm{ K}$ and a density of $0.998 \textrm{ g/cm$^3$}$, with potential energy described by the q-TIP4P/F force field.\cite{Habershon2009}
As in Section~\ref{harmonic_oscillator_results}, we compare the performance of various BAOAB-like T-RPMD schemes for integrating the many-dimensional analogue of Eq.~\ref{eq:pile} with PILE friction, using the simulation time-step $\Delta t = 1.4 \textrm{ fs}$ in all simulations.
Numerical tests reported in Section~\ref{liquid_water_stability} show that this value of $\Delta t$ closely approximates the upper limit of the Verlet (i.e., $n = 1$) stability interval for q-TIP4P/F liquid water.
In agreement with Section~\ref{harmonic_oscillator_results}, the experiments reveal that among the tested T-RPMD schemes, the Cayley scheme offers superior configurational sampling.
For details on the numerical simulation and statistical estimation procedures used to generate the data presented in this section, the reader is referred to Section~\ref{calc_details}.

Figure~\ref{fig:liquid_water_equilibrium_all} compares the equilibrium accuracy achieved by the tested schemes in terms of the quantum and classical kinetic energy per hydrogen atom (Figs.~\ref{fig:liquid_water_equilibrium_all}a, \ref{fig:liquid_water_equilibrium_all}c, and~\ref{fig:liquid_water_equilibrium_all}i) and the intramolecular potential energy per water molecule (Figs.~\ref{fig:liquid_water_equilibrium_all}e and~\ref{fig:liquid_water_equilibrium_all}g); also plotted are the respective integrated autocorrelation times as a function of bead number $n$.
The kinetic energy estimates in Figs.~\ref{fig:liquid_water_equilibrium_all}a and~\ref{fig:liquid_water_equilibrium_all}c exhibit similar trends to those seen in Fig.~\ref{fig:harmonic_oscillator_equilibrium} for the one-dimensional harmonic oscillator.
In particular, the T-RPMD scheme specified by the Cayley angle outperforms others in terms of quantum kinetic energy accuracy as $n$ increases, most outstandingly with a highly accurate primitive kinetic energy estimate despite the large time-step employed.
Still in close agreement with the harmonic oscillator results, Figs.~\ref{fig:liquid_water_equilibrium_all}b and~\ref{fig:liquid_water_equilibrium_all}d show that the Cayley scheme displays the shortest integrated autocorrelation time among the tested schemes for the quantum kinetic energy observables.
Similar trends manifest in the intramolecular potential energy averages and their autocorrelation times (Figs.~\ref{fig:liquid_water_equilibrium_all}e-h), where the Cayley scheme also achieves superior accuracy and efficiency.
Finally, Figs.~\ref{fig:liquid_water_equilibrium_all}i and~\ref{fig:liquid_water_equilibrium_all}j confirm that the relative performance of the compared schemes in terms of velocity-marginal sampling is qualitatively consistent with the harmonic results.
Taken together, the results in Fig.~\ref{fig:liquid_water_equilibrium_all} suggest that the superiority of the Cayley scheme for configurational sampling, proven in the model setting of a harmonic external potential, is also reflected in realistic applications.

In a final numerical test, Fig.~\ref{fig:liquid_water_equilibrium_spec} confirms that the sampling advantages of the Cayley T-RPMD scheme are obtained without downside in the estimation of dynamical quantities of typical interest. 
Specifically, Fig.~\ref{fig:liquid_water_equilibrium_spec}b shows (unnormalized) infrared absorption spectra for room-temperature liquid water, computed from the $128$-bead T-RPMD trajectories used to generate Fig.~\ref{fig:liquid_water_equilibrium_all} using linear response theory and the T-RPMD approximation to real-time quantum dynamics.\cite{Habershon2008,Rossi2014}
Linear response dictates that the absorption spectrum is proportional to $\omega^2 \tilde{\mathcal{I}}(\omega)$, where $\tilde{\mathcal{I}}(\omega) = \int_\mathbb{R} \mathrm{d}t \, e^{-i \omega t} \tilde{C}_{\mu\mu}(t)$ is the Fourier transform of the quantum-mechanical Kubo-transformed dipole autocorrelation function $\tilde{C}_{\mu\mu}(t)$.
The latter is approximated within the T-RPMD framework\cite{Craig2004,Miller2005} by $\tilde{C}_{\mu\mu}(t) \approx \frac{1}{N_\text{H$_2$O}} \sum_{i=1}^{N_\text{H$_2$O}} \mathbb{E} \left( \overline{\mu}_i(t) \cdot \overline{\mu}_i(0) \right)$, where $N_{\text{H$_2$O}}$ is the number of molecules in the liquid, $\overline{\mu}_i(t)$ is the bead-averaged dipole moment of molecule $i$ at time $t$, and the covariance $\mathbb{E} \left( \overline{\mu}_i(t) \cdot \overline{\mu}_i(0) \right)$ is estimated from a stationary T-RPMD trajectory as indicated in Section~\ref{calc_details}.
Figure~\ref{fig:liquid_water_equilibrium_spec}a plots the T-RPMD estimates of $\tilde{C}_{\mu\mu}(t)$ leading to the absorption spectra in Fig.~\ref{fig:liquid_water_equilibrium_spec}b.
On the scale in which the absorption spectrum exhibits its key features, the spectra in Fig.~\ref{fig:liquid_water_equilibrium_spec}b show very minor qualitative discrepancies.
A similar conclusion holds for Fig.~\ref{fig:liquid_water_equilibrium_spec}c, where the T-RPMD approximation of the Kubo-transformed velocity autocovariance function $\tilde{C}_{vv}(t) \approx \frac{1}{N_\text{H$_2$O}} \sum_{i=1}^{N_\text{H$_2$O}} \mathbb{E} \left( \overline{v}_i(t) \cdot \overline{v}_i(0) \right)$ is plotted for the three tested T-RPMD schemes.
Collectively, these observations indicate that the accuracy of dynamical properties computed with BAOAB-like schemes is not significantly affected by the particular $\theta$ employed if conditions~\eqref{eq:C1}-\eqref{eq:C4} in Section~\ref{theory} are met.
This result is expected due to the fact that the considered dynamical properties depend on bead-averaged (i.e., centroid-mode) coordinates, whose evolution is largely independent of the choice of $\theta$ under weak coupling between the centroid and non-centroid ring-polymer modes.

\begin{figure}
    \centering
    \includegraphics[width=0.8\columnwidth]{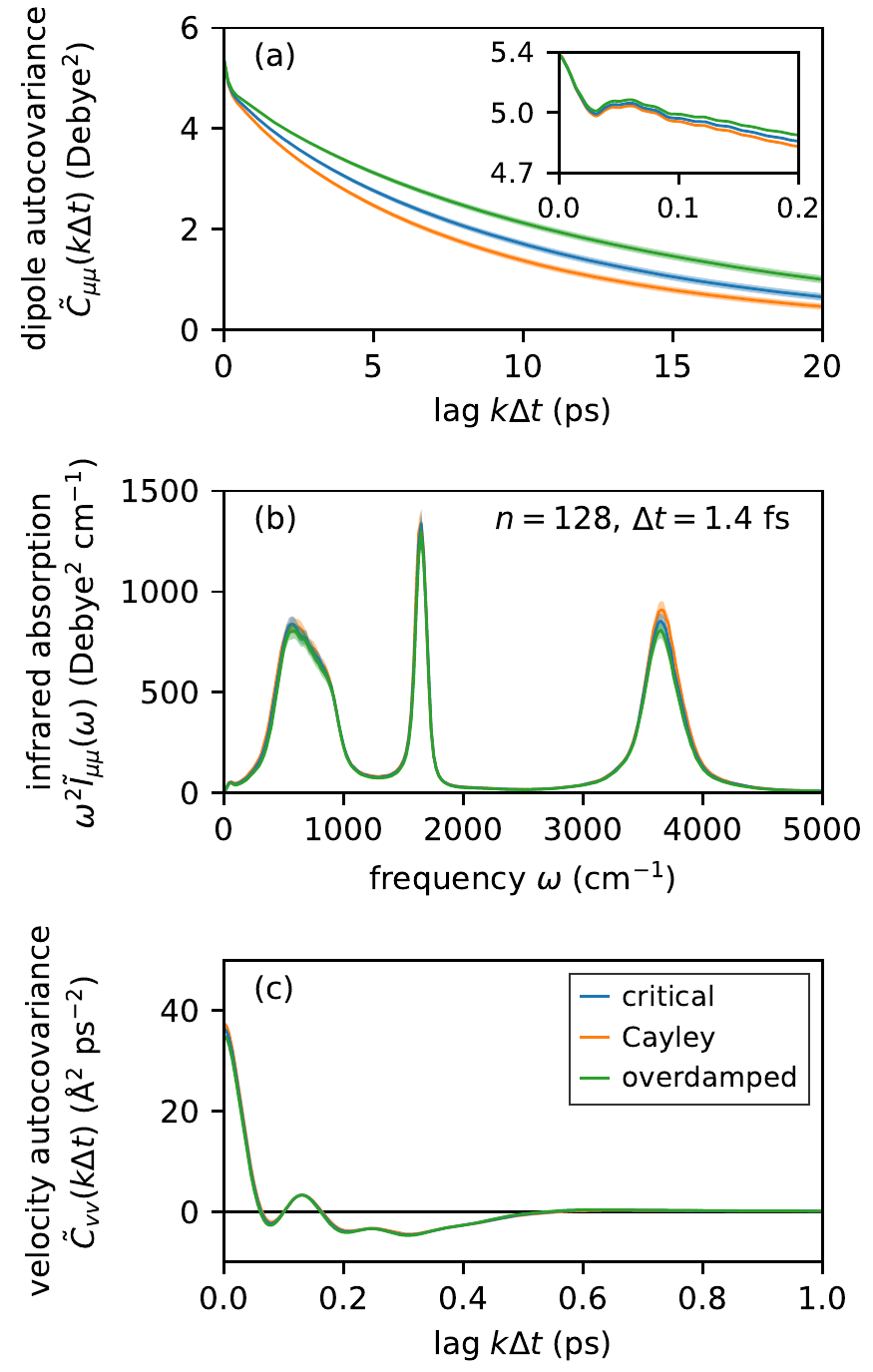}
    \caption{\small 
    Molecular dipole autocovariance function (a), corresponding infrared absorption spectrum (b), and molecular velocity autocovariance function (c) in room-temperature liquid water for various BAOAB-like T-RPMD schemes.
    The plotted quantities autocovariance exhibit minor qualitative discrepancies across schemes, which suggests that all schemes compliant with conditions~\eqref{eq:C1}-\eqref{eq:C4} exhibit comparable accuracy in the computation of dynamical properties.
    Numerical estimates of the autocovariance functions were obtained using the protocol described in Section~\ref{calc_details}.
    }
    \label{fig:liquid_water_equilibrium_spec}
\end{figure}

\section{Summary}
\label{summary}

Previous works showed that strong stability\cite{Korol2019} and dimensionality freedom\cite{Korol2020} are essential features of a robust T-RPMD integration scheme that standard integrators do not possess.
A T-RPMD scheme with these features, denoted BCOCB, was introduced via a simple and inexpensive Cayley modification of the free ring-polymer update (i.e., the ``A'' sub-step) of the standard BAOAB integrator.
The BCOCB scheme was then shown to dramatically outperform BAOAB at estimating static and dynamic properties of various systems with remarkable accuracy at unprecedented time-steps.\cite{Korol2020}

The current work generalizes beyond the Cayley modification by introducing a simple parameterization of the free ring-polymer update and a corresponding family of strongly stable and dimension-free modifications of the BAOAB scheme.
Among these schemes lies BCOCB, which is found to exhibit superior configurational sampling despite exhibiting worse accuracy and efficiency for observables that depend on the non-centroid ring-polymer velocities.
This conclusion is obtained theoretically via exhaustive analysis of a harmonic model, and numerically via simulation of a realistic quantum-mechanical model of liquid water at room temperature.
In this way, the current work convincingly demonstrates the superiority of the BCOCB scheme for accurate and efficient equilibrium simulation of condensed-phase systems with T-RPMD.

To conclude, we stress that implementing BCOCB or any of the new dimension-free and strongly-stable schemes leads to no additional cost, parameters or coding overhead relative to the standard BAOAB integrator.
The modified integrators thus provide ``turnkey'' means to significantly improve the accuracy and stability of existing (T-)RPMD implementations.\cite{Suleimanov2013a,Kapil2019}

\begin{acknowledgements}

J.~L.~R-R.\ and J.~S.\ contributed equally to this work.
This work was supported in part by the U.S.\ Department of Energy (DE-SC0019390) and the National Institutes of Health (R01GM125063).
N.~B.-R.\ acknowledges support by the Alexander von Humboldt foundation and the National Science Foundation (DMS-1816378).

\end{acknowledgements}

\section{Supplementary material}
\label{suppmat}

\subsection{Necessary and sufficient condition for eigenvalues of a $2 \times 2$ real matrix to be inside the unit circle} 
\label{eigs}

This section provides a proof of the standard result used in Sections~\ref{trpmd_free_particle} and~\ref{trpmd_harmonic_oscillator} to infer ergodicity of the T-RPMD update for free and harmonically-confined ring polymers.

\begin{theorem} \label{thm:eigs}
The spectral radius of a $2 \times 2$ real matrix $\m{M}$ is strictly less than one if and only if
\begin{equation} \label{eq:unit_condition}
    |\tr(\m{M})| < 1 + \det(\m{M}) < 2\;.
\end{equation}
\end{theorem}

Fig.~\ref{fig:proof} plots eigenvalue pairs $\lambda_1, \lambda_2$ that satisfy Eq.~\ref{eq:unit_condition} for a fixed value of $\det(\m{M}) = \lambda_1 \lambda_2$.  
Note that the spectral radius of $\m{M}$ is minimized when $\lambda_1$ and $\lambda_2$ are on the circle with radius $r = \sqrt{\det(\m{M})}$.

\begin{proof} 
Let $\lambda_1, \lambda_2$ be the (possibly complex) eigenvalues of $\m{M}$.
By definition, the spectral radius of $\m{M}$ is $\max(|\lambda_1|, |\lambda_2|) =: \rho$.  Since $\m{M}$ is real, both $\tr(\m{M}) = \lambda_1 + \lambda_2$ and $\det(\m{M}) =\lambda_1 \lambda_2$ are real.
Thus, either:
\begin{enumerate}
    \item $\lambda_{1}, \lambda_2$ are a complex conjugate pair; or,
    \item $\lambda_{1}, \lambda_2$ are both real.
\end{enumerate}

In the first case, $\lambda_{1} = a + i b$ and $\lambda_2 = a - i b$ for some real numbers $a$ and $b$ with $b \ne 0$, and hence, $\det(\m{M}) = \lambda_1 \lambda_2 = a^2 + b^2 > 0$, and $\rho = |\lambda_{1}| = |\lambda_{2}| = \sqrt{a^2 + b^2}$, i.e., the eigenvalues lie on the circle with radius $\rho = \sqrt{a^2 + b^2} = \sqrt{\det(\m{M})}$.
In this case, the first inequality in Eq.~\ref{eq:unit_condition} holds since $b \ne 0$ implies
\begin{equation*}
    |\tr(\m{M})| = 2 |a| < 2 \rho \le 1 + \rho^2 = 1 + \det(\m{M}) \;.
\end{equation*}
Hence, Eq.~\ref{eq:unit_condition} is equivalent to $1 + \det(\m{M}) < 2$ or $\rho < 1$.

In the second case, $\lambda_1, \lambda_2$ are both real, and the condition $|\tr(\m{M})| < 1 + \det(\m{M})$ is equivalent to
\begin{equation*}
\begin{aligned}
    & 1 + \lambda_1 \lambda_2 + \lambda_1 + \lambda_2 = (1 + \lambda_1)(1 + \lambda_2) > 0~ \text{and} \;,
    \\
    & 1 + \lambda_1 \lambda_2 - \lambda_1 - \lambda_2 = (1 - \lambda_1)(1 - \lambda_2) > 0\;.
\end{aligned}
\end{equation*}
Together with $\det(\m{M}) = \lambda_1 \lambda_2 < 1$, these conditions are equivalent to $\rho = \max(|\lambda_1|, |\lambda_2|) < 1$. 
\end{proof}

\begin{figure}
    \centering
    \includegraphics[width=0.7\columnwidth]{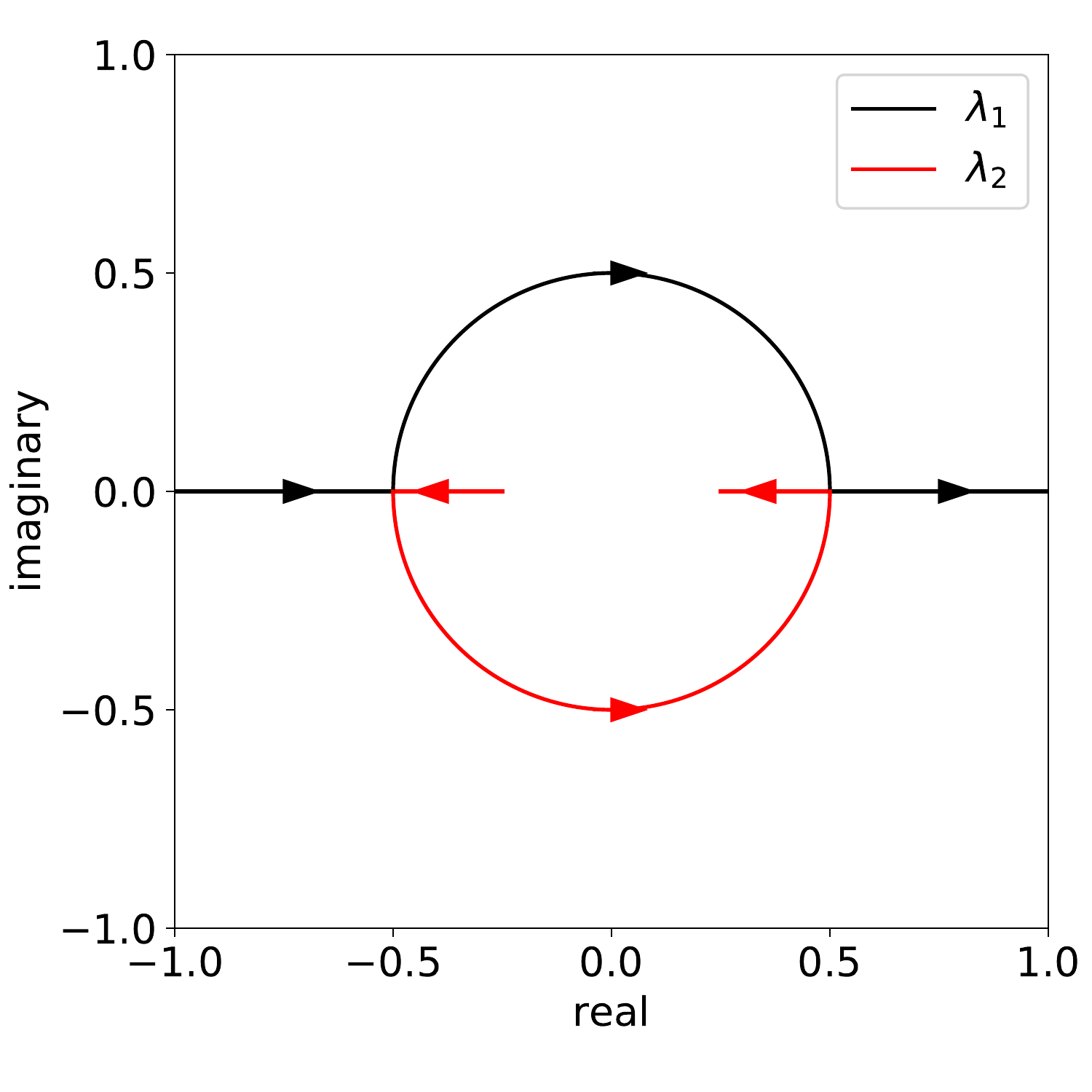}
    \caption{\small
    This figure plots all possible eigenvalue pairs $\lambda_1, \lambda_2$ of a matrix $\m{M}$ that satisfies Eq.~\ref{eq:unit_condition} with $\det(\m{M}) = \lambda_1\lambda_2 = 1/4$.
    The eigenvalue pairs either lie on the circle with radius $r = 1/2$ or are both real, and in the former case, the spectral radius of $\m{M}$ is minimal.
    }
    \label{fig:proof}
\end{figure}

\subsection{Stability condition for harmonic external potentials}
\label{ergodicity}

This section proves that condition~\eqref{eq:C3} implies property~(P3), as claimed in Section~\ref{rpmd_harmonic_oscillator}.
For notational brevity, we define
\begin{equation*}
    A(x) \ := \ \cos(\theta(x)) - \frac{\Delta t^2 (\Lambda/m)}{2} \frac{\sin(\theta(x))}{x} \;.
\end{equation*}
Note that $A(x)$ is equal to $\mathcal{A}_{j,n}$ in the display under Eq.~\ref{eq:tC4} if $x = \omega_{j,n} \Delta t$.

\begin{theorem} \label{thm:ergodicity}
For any $\alpha^{\star} > 0$, (A2) implies (A1).
\begin{description}
    \item[(A1)] For all $\Lambda \ge 0$, $m > 0$ and $\Delta t > 0$ satisfying $\Delta t^2 \Lambda/m < \alpha^{\star}$, the function $\theta$ satisfies
     \begin{equation*}
         \left| A(x) \right| < 1 \quad \text{for $x > 0$} \;.
     \end{equation*}
    \item[(A2)] The function $\theta$ satisfies:
     \begin{equation*}
         0<\theta(x) < 2 \arctan(2 x / \alpha^{\star}) \quad \text{for $x > 0$} \;.
     \end{equation*}
\end{description}
\end{theorem}

\begin{proof}
Let $\alpha = \Delta t^2 (\Lambda/m)$.
For notational brevity, define
\begin{equation*}
    \phi_{\alpha}(x) := \arctan( \alpha / (2 x) )  \quad \text{for $x > 0$} \;.
\end{equation*}
By the harmonic addition identity
\begin{equation*}
    \cos(\theta) - \tan(\phi_{\alpha}) \sin(\theta) =\frac{\cos(\theta + \phi_{\alpha})}{\cos(\phi_{\alpha})} \;,
\end{equation*}
note that (A1) can be rewritten as 
\begin{align} \label{eq:phi_the}
    \left| \frac{\cos (\theta(x) + \phi_{\alpha}(x))}{\cos (\phi_{\alpha}(x))} \right| < 1 \quad \text{for $x > 0$,~ $0 < \alpha < \alpha^{\star}$} \;.
\end{align}
For $0 < \theta(x) < \pi$, Eq.~\ref{eq:phi_the} holds if and only if 
\begin{equation*}
    \phi_{\alpha}(x) < \theta(x) + \phi_{\alpha}(x) < \pi - \phi_{\alpha}(x) \;,
\end{equation*}
which can be rewritten as
\begin{equation} \label{eq:phi_the_2}
    0 < \theta(x) < 2 \arctan(2 x / \alpha) \;,
\end{equation}
where we used the identity
\begin{equation*}
    \pi - 2\arctan(x) = 2\arctan(1/x) \quad \text{valid for $x > 0$} \;.
\end{equation*}
Since $\arctan$ is monotone increasing, and $0<\alpha < \alpha^{\star}$ by assumption, we may conclude that
\begin{equation*}
    0 < \theta(x) < 2 \arctan(2 x / \alpha^{\star})  < 2 \arctan(2 x / \alpha) \;. 
\end{equation*}
Thus, if (A2) holds, then Eq.~\ref{eq:phi_the_2} holds and therefore (A1) holds.
\end{proof}

Fix $\epsilon \in (0, 1)$.
Since Theorem~\ref{ergodicity} is true for arbitrary $\alpha^{\star}$, if we take $\alpha^{\star} = 4 - \epsilon$, then the theorem holds with $\Delta t^2 \Lambda/m < 4 - \epsilon$ in Theorem~\ref{ergodicity} (A1), and $\theta(x) < 2\arctan(2 x/(4 - \epsilon))$ in Theorem~\ref{ergodicity} (A2).
Since $\epsilon > 0$ is arbitrary, and $\arctan$ is monotone increasing, we can conclude that the theorem holds with $\Delta t^2 \Lambda/m < 4$ and $\theta(x) \le 2\arctan(x/2)$.
Summarizing,

\begin{corollary} \label{cor:stab_ho}
Suppose that the function $\theta$ satisfies
\begin{equation*}
    0 < \theta(x) \le 2\arctan(x/2) \quad \text{for $x > 0$} \;.
\end{equation*}
Then for all $\Lambda \ge 0$,  $m > 0$ and $\Delta t > 0$ satisfying $\Delta t^2 \Lambda/m < 4$, we have
\begin{equation*}
    \left| A(x) \right| \ < \ 1 \quad \text{for $x > 0$} \;.
\end{equation*}
\end{corollary}

\subsection{Dimension-free quantitative contraction rate for harmonic external potentials in the infinite-friction limit}
\label{contr}

In the infinite-friction limit, Eq.~\ref{eq:BGOGB_1D_HO} simplifies to
\begin{align*}
    \m{\mathcal{M}}_{j,n} \ &= \ \m{\mathcal{B}}^{1/2} \m{\mathcal{S}}_{j,n}^{1/2} \begin{bmatrix} 1 & 0 \\ 0 & 0 \end{bmatrix} \m{\mathcal{S}}_{j,n}^{1/2} \m{\mathcal{B}}^{1/2} \nonumber \;\; \text{and} 
    \\
    \m{\mathcal{R}}_{j,n} \ &= \ \frac{1}{\beta m_n} \m{\mathcal{B}}^{1/2} \m{\mathcal{S}}_{j,n}^{1/2} \begin{bmatrix} 0 & 0 \\ 0 & 1 \end{bmatrix} ( \m{\mathcal{B}}^{1/2} \m{\mathcal{S}}_{j,n}^{1/2}  )^{\mathrm{T}} \;.
\end{align*} 
The $k$th step of the corresponding T-RPMD integrator can be written compactly as
\begin{equation*}
    \begin{bmatrix} \varrho_j^{(k)} \\ \varphi_j^{(k)} \end{bmatrix}  \ = \ \m{\mathcal{M}}_{j,n} \begin{bmatrix} \varrho_j^{(k-1)} \\ \varphi_j^{(k-1)} \end{bmatrix} + \m{\mathcal{R}}_{j,n}^{1/2} \begin{bmatrix} \xi_j^{(k-1)} \\ \eta_j^{(k-1)} \end{bmatrix} \;,
\end{equation*}
where $\xi_j^{(k-1)}$ and $\eta_j^{(k-1)}$ are independent standard normal random variables.  
Suppose that the initial velocity is drawn from the Maxwell--Boltzmann distribution, i.e., $\varphi_j^{(0)} \sim \mathcal{N}(0, (\beta m_n)^{-1})$ and the initial position is drawn from an arbitrary distribution $\mu_j$ on $\mathbb{R}$, i.e., $\varrho_j^{(0)} \sim \mu_j$.
Let $p^k_{j,n}$ denote the $k$-step transition kernel of the position-marginal, i.e., $\mu_j p^k_{j,n}$ is the probability distribution of $\varrho_j^{(k)}$ with $\varrho_j^{(0)} \sim \mu_j$.

The next theorem shows that starting from any two initial distributions $\mu_j$ and $\nu_j$ on $\mathbb{R}$, the distance between the distributions $\mu_j p_{j,n}^k$ and $\nu_j p_{j,n}^k$ is contractive.
We quantify the distance between these distributions in terms of the $2$-Wasserstein metric.
For two probability distributions $\mu$ and $\nu$ on $\mathbb{R}$, the $2$-Wasserstein distance between $\mu$ and $\nu$ is defined as:
\begin{equation*}
    \mathcal{W}_2(\mu, \nu) = \Big( \inf_{\substack{X \sim \mu \\ Y \sim \nu}} \E( |X - Y|^2 ) \Big)^{1/2} \;,
\end{equation*}
where the infimum is taken over all bivariate random variables $(X,Y)$ such that $X \sim \mu$ and $Y \sim \nu$.\cite{Villani2008}

\begin{theorem} \label{thm:contr}
Suppose that the function $\theta$ satisfies
\begin{equation*}
    0 < \theta(x) \le 2\arctan(x/2) \quad \text{for $x > 0$} \;. 
\end{equation*}
Then for all $k>1$, $\Lambda \ge 0$,  $m > 0$ and $\Delta t > 0$ satisfying $\Delta t^2 \Lambda/m < 4$, and for all initial distributions $\mu_j$ and $\nu_j$ on $\mathbb{R}$,
\begin{equation}
\begin{aligned}
   & \mathcal{W}_2(\mu_j p^k_{j,n}, \nu_j p^k_{j,n}) \ \le \\ 
   &  \begin{cases} A(\omega_{j,n} \Delta t)^{k-1} \mathcal{W}_2(\mu_j, \nu_j) & \text{if $A(\omega_{j,n} \Delta t)>0$}, \\
    \frac{1}{2} \frac{1}{k-1} \mathcal{W}_2(\mu_j, \nu_j) & \text{else}. \end{cases} \\
\end{aligned}
\end{equation}
\end{theorem}

\begin{proof}
In the infinite-friction limit, the eigenvalues of $\m{\mathcal{M}}_{j,n}$ are $\{0, A(\omega_{j,n} \Delta t)\}$, where $A(x)$ is defined in Section~\ref{ergodicity}.
Let $\varrho_j^{(0)} \sim \mu_j$ and $\tilde \varrho_j^{(0)} \sim \nu_j$ be an optimal coupling of $\mu_j$ and $\nu_j$, i.e., $\mathcal{W}_2(\mu_j, \nu_j) = \E(| \varrho_j^{(0)} - \tilde \varrho_j^{(0)} |^2)^{1/2}$.
Conditional on $\varrho_j^{(0)}$ and $\tilde \varrho_j^{(0)}$, $\varrho_j^{(k)}$ and $\tilde \varrho_j^{(k)}$ are Gaussian random variables with equal variances, but different means.
By a well-known result for the $2$-Wasserstein distance between Gaussian distributions,\cite{Givens1984}
\begin{align}
    &\mathcal{W}_2(\mu_j p^k_{j,n}, \nu_j p^k_{j,n})^2 \nonumber \\
    &= |A(\omega_{j,n} \Delta t)|^{2(k-1)} (\m{\mathcal{M}}_{j,n})_{11}^2 \mathcal{W}_2(\mu_j , \nu_j)^2 \nonumber \\
    &= |A(\omega_{j,n} \Delta t)|^{2(k-1)} \frac{( 1 + A(\omega_{j,n} \Delta t) )^2}{4} \mathcal{W}_2(\mu_j , \nu_j)^2 \;, \label{w2:identity}
\end{align}
where we used $(\m{\mathcal{M}}_{j,n})_{11} = ( 1 + A(\omega_{j,n} \Delta t) )/2$.

Now we distinguish between two cases.
In the case where $A(\omega_{j,n} \Delta t)>0$, we obtain the required result since $|A(\omega_{j,n} \Delta t)| < 1$ by Corollary~\ref{cor:stab_ho}, and therefore,
\begin{equation} \label{case:1}
    \frac{( 1 + A(\omega_{j,n} \Delta t) )^2}{4} \le 1 \;.
\end{equation}
Otherwise, for $-1 < A(\omega_{j,n} \Delta t) \le 0$ the quantity $|A(\omega_{j,n} \Delta t)|^{2(k-1)} ( 1 + A(\omega_{j,n} \Delta t) )^2$ is maximized at $(-1+1/k)^{2 k} (k-1)^{-2}$, and therefore,
\begin{equation} \label{case:2}
    |A(\omega_{j,n} \Delta t)|^{2(k-1)} \frac{( 1 + A(\omega_{j,n} \Delta t) )^2}{4}  \le \frac{1}{4 (k-1)^2} \;.
\end{equation}
Inserting Eq.~\ref{case:1} and Eq.~\ref{case:2} into Eq.~\ref{w2:identity}, and then taking square roots, gives the required result.
\end{proof}

\subsection{Total variation bound on the equilibrium accuracy error for harmonic external potentials}
\label{tv_error_proof}

In this section, we show that Eq.~\ref{eq:tv_error} follows from conditions~\eqref{eq:C1}-\eqref{eq:C4} in the setting of Section~\ref{trpmd_harmonic_oscillator}.
It is helpful to recall the quantities
\begin{equation} \label{eq:continuous_eigenvalues}
    \omega_{j} = \lim_{n \to \infty} \omega_{j,n} =
    \begin{cases}
    \frac{\pi j}{\hbar \beta} & \text{if $j$ is even} \;, \\
    \frac{\pi (j+1)}{\hbar \beta} & \text{else} \;.
    \end{cases}
\end{equation}
In the following, $\mu_{j, \Delta t}$ and $\mu_{j}$ respectively denote the $j$th factor of the product distributions $\mu_{n, \Delta t}$ and $\mu_{n}$ introduced in Section~\ref{trpmd_harmonic_oscillator}.

\begin{theorem} \label{thm:tv_error}
Suppose that the function $\theta$ satisfies conditions~\eqref{eq:C1}-\eqref{eq:C4}.
Then for all $\Lambda \ge 0$,  $m > 0$ and $\Delta t > 0$ satisfying $\Delta t^2 \Lambda/m < 4$, the total variation distance between $\mu_{n}$ and $\mu_{n, \Delta t}$ is bounded as in Eq.~\ref{eq:tv_error}.
\end{theorem}

\begin{proof}
Subadditivity of the total variation distance $d_{\mathrm{TV}}$ between product distributions and its equivalence with the Hellinger distance\cite{Gibbs2002} $d_{\mathrm{H}}$ lead to the inequalities
\begin{align} \label{eq:tv_error_ineqs}
    d_{\mathrm{TV}} & (\mu_{n}, \mu_{n, \Delta t})^2  \le \sum_{j=1}^{n-1} d_{\mathrm{TV}}(\mu_{j}, \mu_{j, \Delta t})^2 \nonumber \\
    &\le \sum_{j=1}^{n-1} 2 \, d_{\mathrm{H}}(\mu_{j}, \mu_{j, \Delta t})^2  \le \sum_{j=1}^{n-1} \frac{ 2 ( s_{j} - s_{j, \Delta t} )^2 }{ ( s_{j}^2 + s_{j, \Delta t}^2 ) } \nonumber \\
    &\le \sum_{j=1}^{n-1} \left( 1 - \frac{ s_{j} }{ s_{j, \Delta t} } \right)^{\!2} \le \sum_{j=1}^{n-1} \left( 1 - \frac{ s_{j}^2 }{ s_{j, \Delta t}^2 } \right)^{\!2} \;,
\end{align}
where the second-to-last step uses Eq.~\ref{eq:exact_numerical_IM_1D_comparison} and the last step uses the elementary inequality $(1-x^2)^2 \ge (1-x)^2$ valid for all $x \ge 0$.

Since $\tan(\cdot)$ increases superlinearly on the interval $(0, \pi)$, we have $\theta(x)/2 \le \tan(\theta(x)/2) \le x/2$ for $x > 0$, where the second inequality uses \eqref{eq:C3}.
Consequently, the $j$th summand in Eq.~\ref{eq:tv_error_ineqs} admits the bound
\begin{equation*}
\begin{aligned}
    \left( 1 - \frac{ s_{j}^2 }{ s_{j, \Delta t}^2 } \right)^{\!2}
    &= \left( \frac{ \Lambda/m }{ \omega_{j,n}^2 + \Lambda/m } \left( \frac{\omega_{j,n} \Delta t / 2} {\tan \left( \theta(\omega_{j,n} \Delta t) / 2 \right)} - 1 \right) \right)^{\!2} \\
    &\le \left( \frac{\Delta t^2 \Lambda / m}{ (\omega_{j, n} \Delta t)^2} \left( \frac{\omega_{j,n} \Delta t}{\theta(\omega_{j,n} \Delta t)} - 1 \right) \right)^{\!2} \\
    &\le \left( \frac{\Delta t^2 \Lambda}{m} \right)^{\!2} \frac{1}{(\omega_{j,n} \Delta t)^{2}} \;,
\end{aligned}
\end{equation*}
where the last line uses the lower bound in \eqref{eq:C4}.
Using that for any even positive integer $n$
\begin{equation*}
    \sum_{j=1}^{n-1} \frac{1}{\omega_{j,n}^{2}}
    < \lim_{n \to \infty} \sum_{j=1}^{n-1} \frac{1}{\omega_{j,n}^{2}}
    = \sum_{j=1}^{\infty} \frac{1}{\omega_{j}^{2}}
    < \left( \frac{ \hbar \beta } { \pi } \right)^{\!2} \sum_{j=1}^{\infty} \frac{2}{j^2} \;,
\end{equation*}
where we used Eq.~\ref{eq:continuous_eigenvalues}, the bound in Eq.~\ref{eq:tv_error_ineqs} becomes
\begin{equation*}
\begin{aligned}
    d_{\mathrm{TV}}(\mu_{n}, \mu_{n, \Delta t})^2
    &< \left( \frac{\Delta t^2 \Lambda}{m} \right)^{\!2} \left( \frac{ \hbar \beta } { \pi \Delta t } \right)^{\!2} \sum_{j=1}^{\infty} \frac{2}{j^2} \;.
\end{aligned}
\end{equation*}
Taking square roots and using the Riemann zeta function\cite{Abramovitz1965} to evaluate the infinite sum yields Eq.~\ref{eq:tv_error}.
\end{proof}

\subsection{Asymptotic variance of kinetic energy observables for harmonic external potentials in the infinite-friction limit}
\label{analytical_asymptotic_variance}

In Section~\ref{harmonic_oscillator_results}, Figs.~\ref{fig:harmonic_oscillator_equilibrium}b and~\ref{fig:harmonic_oscillator_equilibrium}d show that the T-RPMD scheme specified by $\theta(x) = 2\arctan(x/2)$, which coincides with the Cayley-modified BAOAB scheme introduced in Ref.~\onlinecite{Korol2020}, provides the smallest integrated autocorrelation time (Eq.~\ref{eq:obs_iact}) for quantum kinetic energy observables (Eq.~\ref{eq:qke}) among several schemes with properties~(P1)-(P5). 
In this section, we show that this scheme minimizes an upper bound (Eq.~\ref{eq:qke_iact_harmonic_gammainf_upperbnd}) on the integrated autocorrelation time of the quantum kinetic energy among all dimension-free and strongly-stable BAOAB-like schemes for harmonic external potentials.

To this end, note that for a $n$-bead thermostatted ring polymer with external potential $V_n^\textrm{ext}(\vec{q}) = \frac{\Lambda}{2n} | \vec{q} |^2$, Eq.~\ref{eq:qke} can be rewritten
\begin{equation} \label{eq:qke_nm_harmonic}
\begin{aligned}
    \mathsf{KE}^\mathrm{pri}_n(\vec{\varrho}) &= \frac{n}{2 \beta} - \sum_{j=1}^{n-1} \frac{m_n \omega_{j,n}^2}{2} \varrho_j^2 \;\; \text{and} \\
    \mathsf{KE}^\mathrm{vir}_n(\vec{\varrho}) &= \frac{1}{2 \beta} + \sum_{j=1}^{n-1} \frac{\Lambda}{2n} \varrho_j^2 \;
\end{aligned}
\end{equation}
where $\vec{\varrho}$ is defined in Eq.~\ref{eq:nm_transform}.
In the following, we denote both observables in Eq.~\ref{eq:qke_nm_harmonic} as $\mathsf{KE}_n$ and distinguish between the two as needed.

To control the integrated autocorrelation time of $\mathsf{KE}_n$, we need the stationary autocorrelation $\mathrm{Cor}\big( \mathsf{KE}_n(\vec{\varrho}^{(0)}), \mathsf{KE}_n(\vec{\varrho}^{(k \Delta t)}) \big)$ for $k \ge 0$.
Note that the distributions of $\vec{\varrho}^{(k \Delta t)}$ and $\vec{\varrho}^{(0)}$ are equal by stationarity, and that components $( \varrho_j )_{j=0}^{n-1}$ are uncorrelated in a harmonic external potential.
Thus,
\begin{align*}
    \mathrm{Cor} \big( \mathsf{KE}_n(\vec{\varrho}^{(0)}), & \, \mathsf{KE}_n(\vec{\varrho}^{(k\Delta t)}) \big) \\
    &= \sum_{j=1}^{n-1} \chi_{j,n} \mathrm{Cor} \big( |\varrho_j^{(0)}|^2 , |\varrho_j^{(k \Delta t)}|^2 \big) \;,
\end{align*}
where
\begin{equation*}
    \chi_{j,n} = \frac{\kappa_{j,n}^2 \mathrm{Var} \big( |\varrho_j^{(0)}|^2 \big)}{\sum_{i=1}^{n-1} \kappa_{i,n}^2 \mathrm{Var} \big( |\varrho_i^{(0)}|^2 \big)} 
\end{equation*}
and
\begin{equation*}
    \kappa_{j,n} =
    \begin{cases}
    \frac{m_n \omega_{j,n}^2}{2} &\text{ for $\mathsf{KE}_n^\mathrm{pri}$} \;, \\
    \frac{\Lambda}{2 n} &\text{ for $\mathsf{KE}_n^\mathrm{vir}$} \;.
    \end{cases}
\end{equation*}
If the evolution of the ring polymer is governed by the BAOAB-like update in Eq.~\ref{eq:BGOGB_1D}, then the $j$th mode satisfies
\begin{equation*}
\begin{aligned}
    \mathrm{Cor} \big( |\varrho_j^{(0)}|^2, |\varrho_j^{(k \Delta t)}|^2 \big)
    &= \frac{ 
    \mathrm{Cov} \big( |\varrho_j^{(0)}|^2, |\varrho_j^{(k \Delta t)}|^2 \big)
    }{
    \mathrm{Var} \big( |\varrho_j^{(0)}|^2 \big)
    } \\
    &= ( \m{\mathcal{M}}_{j,n}^k )_{11}^2 \;,
\end{aligned}
\end{equation*}
where we used that the phase $\begin{bmatrix} \varrho_j^{(k \Delta t)} & \varphi_j^{(k \Delta t)} \end{bmatrix}^\mathrm{T}$ follows a centered Gaussian distribution with covariance given in Eq.~\ref{eq:numerical_IM_1D_a} for all $k \ge 0$.
Therefore, in the infinite-friction limit where $\m{\mathcal{M}}_{j,n}$ is given in Section~\ref{contr}, the integrated autocorrelation time of $\mathsf{KE}_n$ evaluates to
\begin{align} \label{eq:qke_iact_harmonic_gammainf_upperbnd}
    \frac{\mathrm{aVar}(\mathsf{KE}_n)}{\mathrm{Var}(\mathsf{KE}_n)}
    &= 1 + 2 \sum_{j=1}^{n-1} \chi_{j,n} \sum_{k=1}^\infty (\m{\mathcal{M}}_{j,n}^k)_{11}^2 \nonumber \\
    &\le 1 + \frac{1}{2} \max_{1 \le j \le n-1} \left| \frac{1 + A(\omega_{j,n} \Delta t)}{1 - A(\omega_{j,n} \Delta t)} \right| \;,
\end{align}
where simplification of $( \m{\mathcal{M}}_{j,n}^k )_{11}$ was aided by the Cayley--Hamilton theorem for $2 \times 2$ matrices,\cite{Andreescu2016} $A(x)$ is defined in Section~\ref{ergodicity}, and in the last line we used that $\sum_{j=1}^{n-1} \chi_{j,n} = 1$.
Eq.~\ref{eq:qke_iact_harmonic_gammainf_upperbnd} states that the integrated autocorrelation time of $\mathsf{KE}_n$ can only be as small as that of the component $| \varrho_j |^2$ exhibiting the slowest uncorrelation at stationarity.

Having derived Eq.~\ref{eq:qke_iact_harmonic_gammainf_upperbnd}, we now prove our claim for this section.
Let $x := \omega_{j,n} \Delta t > 0$ and $\alpha := \Delta t^2 \Lambda / m \in (0, 4)$.
For fixed $x$ and $\alpha$, the function $A(x) := \cos(\theta(x)) - \frac{\alpha}{2x} \sin(\theta(x))$ monotonically decreases toward $-1$ as the angle $\theta(x)$ increases toward $\pi$.
Consequently, the function $\left| \big( 1 + A(x) \big) / \big( 1 - A(x) \big) \right|$ decreases (toward $0$) as $\theta(x)$ increases (toward $\pi$), but condition~\eqref{eq:C3} requires $\theta(x) \le 2\arctan(x/2)$ to achieve stable evolution.
Therefore, because it yields the largest stable angle, the choice $\theta(x) = 2\arctan(x/2)$ (i.e., the Cayley angle) minimizes the upper bound in Eq.~\ref{eq:qke_iact_harmonic_gammainf_upperbnd}.

A similar argument can be made to support the conjecture, suggested by Fig.~\ref{fig:harmonic_oscillator_equilibrium}f, that the non-centroid velocity estimator for the classical kinetic energy $\mathsf{KE}_n^\textrm{cla}$ in Eq.~\ref{eq:cke}, equivalently written
\begin{equation} \label{eq:cke_nm_harmonic}
    \mathsf{KE}_n^\mathrm{cla}(\vec{\varphi}) = \frac{m_n}{2(n-1)} \sum_{j=1}^{n-1} \varphi_j^2
\end{equation}
with $\vec{\varphi}$ defined in Eq.~\ref{eq:nm_transform}, exhibits a \textit{maximal} integrated autocorrelation time if the Cayley angle $\theta(x) = 2\arctan(x/2)$ is used.
Indeed, the integrated autocorrelation time of this estimator is bounded by 
\begin{equation} \label{eq:cke_iact_harmonic_gammainf_upperbnd}
    \frac{\mathrm{aVar}(\mathsf{KE}_n^\textrm{cla})}{\mathrm{Var}(\mathsf{KE}_n^\textrm{cla})}
    \le 1 + \frac{1}{2} \max_{1 \le j \le n-1} \left| \frac{1 - A(\omega_{j,n} \Delta t)}{1 + A(\omega_{j,n} \Delta t)} \right| \;,
\end{equation}
where the function $\left| \big( 1 - A(x) \big) / \big( 1 + A(x) \big) \right|$ is maximized as $\theta(x)$ approaches the largest stable (i.e., Cayley) angle for fixed $x$ and $\alpha$.

To conclude, we note that the conclusions of this section hold for arbitrary friction schedules despite our use of the infinite-friction limit in Eqs.~\ref{eq:qke_iact_harmonic_gammainf_upperbnd} and~\ref{eq:cke_iact_harmonic_gammainf_upperbnd}.

\subsection{Stability interval calibration for liquid water simulations}
\label{liquid_water_stability}

This section describes the computational procedure used to identify $\Delta t = 1.4 \textrm{ fs}$ as close to the upper bound of the stability interval of T-RPMD applied to q-TIP4P/F liquid water at $298 \textrm{ K}$ and $0.998 \textrm{ g/cm$^3$}$.
The procedure consisted of integrating an ensemble of $10^4$ thermally initialized T-RPMD trajectories using the algorithm outlined in Section~\ref{trpmd} in its single-bead realization (identical to velocity Verlet in classical MD\cite{Leimkuhler2015}), and counting the fraction of trajectories that remained within an energy sublevel (i.e., did not exhibit detectable energy drift) throughout their duration for each tested time-step.
A time-step was deemed stable if $99\%$ or more of the ensemble remained in an energy sublevel throughout a $50$-picosecond time period.
A range of time-steps was tested, and the fraction of stable trajectories at each time-step is reported in Fig.~\ref{fig:liquid_water_stabfrac}.

To avoid initialization bias in the stability interval estimation, thermalized initial phase-points were generated with a Metropolized Markov-chain Monte Carlo sampler targeted at the equilibrium configurational distribution of the liquid.
Specifically, a randomized Hamiltonian Monte Carlo\cite{BouRabee2017a,BouRabee2018} (rHMC) simulation of sufficient length was used to thermalize a crystalline configuration of the system at the target density, and $10^2$ configurations were extracted from well-separated points along the rHMC trajectory.
Each of these approximately independent draws from the equilibrium \textit{configurational} distribution of the liquid at the target physical conditions was subsequently paired with $10^2$ independent velocities drawn from the corresponding Maxwell--Boltzmann distribution, yielding $10^4$ approximately independent draws from the \textit{phase space} distribution of the classical liquid at thermal equilibrium.

\begin{figure}
    \centering
    \includegraphics[width=0.6\columnwidth]{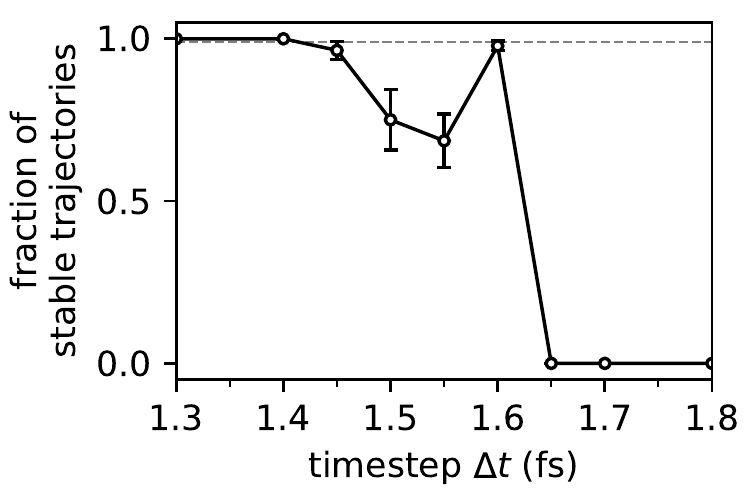}
    \caption{\small 
    Stability interval calibration for q-TIP4P/F room-temperature liquid water simulations.
    Data points correspond to the fraction of thermally initialized single-bead T-RPMD trajectories that remained stable over a $50$-picosecond period at the respective integration time-step $\Delta t$.
    Error bars correspond to the standard error of the fraction of stable trajectories across initialization points with different configurations.
    The gray dashed line marks the $\ge\!99\%$ threshold for deeming a time-step \textit{stable},
    which no time-step beyond $\Delta t = 1.4 \textrm{ fs}$ reaches.
    }
    \label{fig:liquid_water_stabfrac}
\end{figure}

\subsection{Simulation and estimation details}
\label{calc_details}

This section compiles simulation protocols and statistical estimation methods used to generate Figs.~\ref{fig:harmonic_oscillator_equilibrium} for the one-dimensional quantum harmonic oscillator, and Figs.~\ref{fig:liquid_water_equilibrium_all} and~\ref{fig:liquid_water_equilibrium_spec} for room-temperature liquid water.

\subsubsection{One-dimensional quantum harmonic oscillator}

Numerical equilibrium averages and integrated autocorrelation times for the quantum harmonic oscillator were estimated by averaging over a $10$-nanosecond T-RPMD trajectory integrated using the algorithm listed in Section~\ref{trpmd}, and initialized at an exact sample from the numerical stationary distribution (listed for the $j$th ring-polymer mode in Eq.~\ref{eq:numerical_IM_1D_b}) corresponding to the physical parameters (i.e., $\Lambda$, $m$, and $\beta$) and simulation parameters (i.e., $n$, $\Delta t$, and the function $\theta$) listed in Section~\ref{harmonic_oscillator_results}.
Specifically, the statistics reported in Fig.~\ref{fig:harmonic_oscillator_equilibrium} were obtained by partitioning the T-RPMD trajectory into $10$ disjoint blocks, estimating the equilibrium average and autocorrelation time within each block, and computing the sample mean and standard error among the resulting block estimates with $1000$ bootstrap resamples.

We now describe the formulas and methods used to obtain block estimates for the equilibrium mean and integrated autocorrelation time.
The equilibrium average $\mu_{\mathsf{O}_n}$ of observable $\mathsf{O}_n$ within each block of the partitioned T-RPMD trajectory was estimated using the standard estimator\cite{Priestly1981}
\begin{equation} \label{eq:mean_obs_estimator}
    \hat{\mu}_{\mathsf{O}_n} = \frac{1}{K} \sum_{k=0}^{K-1} \mathsf{O}_n^{(k \Delta t)} \;,
\end{equation}
where $K$ is the number of steps in the block (i.e., the block size) and $\mathsf{O}_n^{(k \Delta t)}$ the value of $\mathsf{O}_n$ at the $k$th step within the block.
Similarly, the lag-$k \Delta t$ autocovariance $C_{\mathsf{O}_n}(k \Delta t)$ was estimated using\cite{Priestly1981}
\begin{equation*}
    \hat{C}_{\mathsf{O}_n}(k \Delta t) = \sum_{\ell=0}^{K-k-1} \frac{ \big( \mathsf{O}_n^{(\ell \Delta t)} \!-\! \hat{\mu}_{\mathsf{O}_n} \big) \big( \mathsf{O}_n^{\left((\ell + k) \Delta t\right)} \!-\! \hat{\mu}_{\mathsf{O}_n} \big) }{ K - k }
\end{equation*}
for $0 \le k \Delta t \le (K-1) \Delta t = 1 \textrm{ ns}$.
The integrated autocorrelation time was subsequently estimated using\cite{Priestly1981,Sokal1997}
\begin{equation} \label{eq:iact_obs_estimator}
    \widehat{ \frac{ \mathrm{aVar}_{\mathsf{O}_n} }{ \mathrm{Var}_{\mathsf{O}_n} } }(M) = 1 + 2 \sum_{k=1}^M \frac{ \hat{C}_{\mathsf{O}_n} ( k \Delta t ) }{ \hat{C}_{\mathsf{O}_n} ( 0 ) } \;,
\end{equation}
where $0 < M \le K$ is a suitable cutoff.
The choice of $M$ is nontrivial, as it carries a trade-off between bias (more pronounced at small $M$) and variance (more pronounced at large $M$).\cite{Sokal1997}
To choose $M$ judiciously, we follow the \textit{automatic windowing} (AW) method described in Appendix C of Ref.~\onlinecite{Madras1988}.
The AW method dictates that $M$ should correspond to the smallest lag that satisfies the inequality
\begin{equation*}
    M \ge c \, \widehat{ \frac{ \mathrm{aVar}_{\mathsf{O}_n} }{ \mathrm{Var}_{\mathsf{O}_n}} }(M) \;,
\end{equation*}
where the parameter $c > 0$ dictates the variance-bias trade-off in place of $M$, and is chosen as large as possible to reduce the bias of the estimator for a given variance threshold.

Fig.~\ref{fig:harmonic_oscillator_awdemo} illustrates usage of the AW method for integrated autocorrelation time estimation, using trajectory data generated by the T-RPMD scheme with $\theta(x) = 2\arctan(x/2)$ at $n = 64$ beads and $\Delta t = 2.0 \textrm{ fs}$, and focusing on the observables $\mathsf{KE}_n^\textrm{pri}$ (black), $\mathsf{KE}_n^\textrm{vir}$ (red), and $\mathsf{KE}_n^\textrm{cla}$ (cyan) introduced in Section~\ref{harmonic_oscillator_results}.
The estimated integrated autocorrelation times are plotted with solid lines in Fig.~\ref{fig:harmonic_oscillator_awdemo}a for various values of $c$, and the corresponding cutoffs $M$ are plotted in Fig.~\ref{fig:harmonic_oscillator_awdemo}b.
Exact integrated autocorrelation times are plotted with dashed lines in Fig.~\ref{fig:harmonic_oscillator_awdemo}a.
Note that as $c$ (and thus $M$) increases, the estimates converge to the corresponding exact values at the expense of a larger variance, which can nonetheless be controlled by adjusting the block size $K$.

\begin{figure}
    \centering
    \includegraphics[width=0.6\columnwidth]{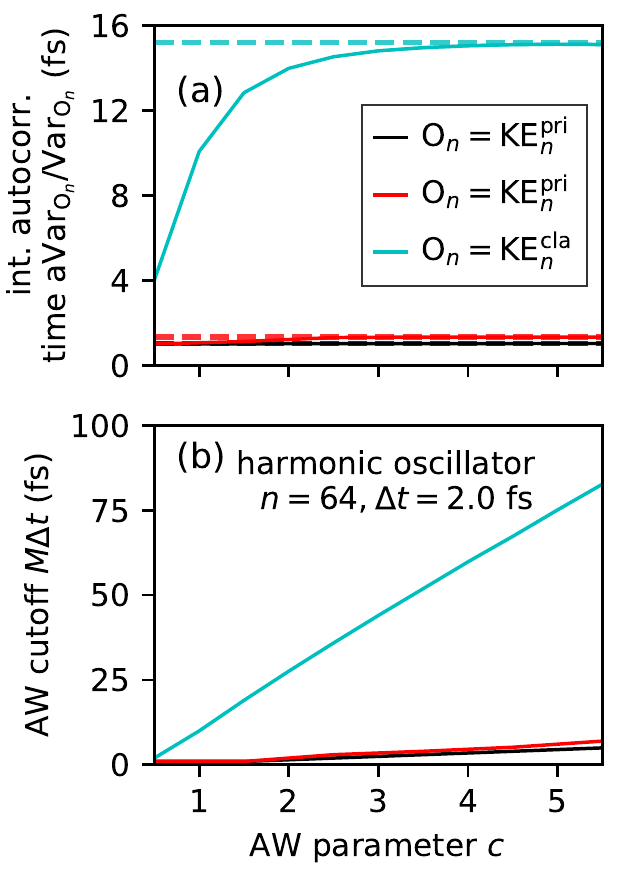}
    \caption{\small 
    Integrated autocorrelation times of several observables of the one-dimensional harmonic oscillator in Section~\ref{harmonic_oscillator_results}, estimated with the AW method.
    Trajectory data for the estimates was generated using the T-RPMD scheme with $\theta(x) = 2\arctan(x/2)$ at $n = 64$ beads and $\Delta t = 2.0 \textrm{ fs}$, and processed as described in the current section.
    Estimated (resp.\ exact) integrated autocorrelation times for observables $\mathsf{KE}_n^\textrm{pri}$ (black), $\mathsf{KE}_n^\textrm{vir}$ (red), and $\mathsf{KE}_n^\textrm{cla}$ (cyan) are shown in solid (resp.\ dashed) lines in panel (a) as a function of the windowing parameter $c$.
    Panel (b) plots the cutoffs determined by the choice of $c$ for the three observables, where the linear relation between $M \Delta t$ and $c$ at large values of the latter corroborates the non-spurious convergence of the autocorrelation time estimates.
    }
    \label{fig:harmonic_oscillator_awdemo}
\end{figure}

\subsubsection{Room-temperature liquid water}

The equilibrium averages and integrated autocorrelation times reported in Fig.~\ref{fig:liquid_water_equilibrium_all} were obtained by averaging over $10$-nanosecond T-RPMD trajectories integrated for each considered bead number $n$, time-step $\Delta t$, and function $\theta$.
All trajectories were initialized at an approximate sample from the corresponding numerical equilibrium distribution, obtained by thermalizing for $20$ picoseconds a classical (i.e., $n = 1$) configuration of the system into the $n$-bead ring-polymer phase space.
The reference equilibrium averages plotted with dashed lines in Fig.~\ref{fig:liquid_water_equilibrium_all} were obtained by averaging over a one-nanosecond, $256$-bead
staging PIMD\cite{Tuckerman1993} trajectory integrated at a $0.1$-fs time-step with the mass and friction parameters recommended in Ref.~\onlinecite{Liu2016}, and initialized with the same protocol used for the T-RPMD simulations.

The observables considered in Fig.~\ref{fig:liquid_water_equilibrium_all} measure properties per $\mathrm{H}$ atom or per $\mathrm{H}_2\mathrm{O}$ molecule, and thus the reported values are averages over estimates obtained for each simulated moiety.
The equilibrium mean and integrated autocorrelation time of observable $\mathsf{O}_n$ for each moiety was estimated by partitioning the trajectory of the moiety into $10$ disjoint $1$-nanosecond blocks, evaluating Eqs.~\ref{eq:mean_obs_estimator} and~\ref{eq:iact_obs_estimator} within each block, and determining the sample mean and standard error among the block estimates with $1000$ bootstrap resamples.
The AW method\cite{Madras1988} was applied to choose a cutoff lag $M \le 1 \textrm{ ns}$ in Eq.~\ref{eq:iact_obs_estimator}, as illustrated in Fig.~\ref{fig:harmonic_oscillator_awdemo} for the harmonic oscillator application.

The T-RPMD trajectories used to generate Fig.~\ref{fig:liquid_water_equilibrium_all} also yielded Fig.~\ref{fig:liquid_water_equilibrium_spec}, where panels~(a) and~(c) plot autocovariance functions of the form $\frac{1}{N_{\mathrm{H}_2\mathrm{O}}} \sum_{i=1}^{N_{\mathrm{H}_2\mathrm{O}}} \mathbb{E} \big( \bar{\mathsf{O}}_i(0) \cdot \bar{\mathsf{O}}_i(k \Delta t) \big)$, where $N_{\mathrm{H}_2\mathrm{O}} = 32$ is the number of simulated $\mathrm{H}_2\mathrm{O}$ molecules and $\bar{\mathsf{O}}_i(k \Delta t)$ is the bead-averaged value of observable $\mathsf{O}$ (e.g., the molecular dipole moment or center-of-mass velocity) on the $i$th molecule at time $k \Delta t$ along a stationary T-RPMD trajectory.
The autocovariance $\mathbb{E} \big( \bar{\mathsf{O}}_i(0) \cdot \bar{\mathsf{O}}_i(t) \big)$ was estimated for the lags $k \Delta t$ shown in Fig.~\ref{fig:liquid_water_equilibrium_spec} by
\begin{equation*}
    \mathbb{E} \big( \bar{\mathsf{O}}_i(0) \cdot \bar{\mathsf{O}}_i(k \Delta t) \big) \\
    \approx \sum_{\ell=0}^{K-k-1} \frac{\bar{\mathsf{O}}_i^{(\ell \Delta t)} \cdot \bar{\mathsf{O}}_i^{\left((\ell + k) \Delta t\right)}}{K-k} \;,
\end{equation*}
where $K \Delta t = 1 \textrm{ ns}$ is the length of each block in the partitioned $10$-nanosecond T-RPMD trajectory.
As with the results in Fig.~\ref{fig:liquid_water_equilibrium_all}, autocovariance statistics for each molecule were obtained from block estimates via bootstrapping, and Figs.~\ref{fig:liquid_water_equilibrium_spec}a and~\ref{fig:liquid_water_equilibrium_spec}c report molecule-averaged statistics.

Fig.~\ref{fig:liquid_water_convergence} validates the $20$-picosecond thermalization interval used to initialize the trajectories that generated Figs.~\ref{fig:liquid_water_equilibrium_all} and~\ref{fig:liquid_water_equilibrium_spec}.
In detail, Figs.~\ref{fig:liquid_water_convergence}a and~\ref{fig:liquid_water_convergence}b (resp., Figs.~\ref{fig:liquid_water_convergence}c and~\ref{fig:liquid_water_convergence}d) plot the non-equilibrium mean of the primitive and virial quantum kinetic energy per $\mathrm{H}$ atom (resp.\ the mean $\mathrm{O}\!-\!\mathrm{H}$ bond and $\mathrm{H}\!-\!\mathrm{O}\!-\!\mathrm{H}$ angle potential energy per water molecule) as it approaches the equilibrium value in Figs.~\ref{fig:liquid_water_equilibrium_all}a and~\ref{fig:liquid_water_equilibrium_all}c (resp., Figs.~\ref{fig:liquid_water_equilibrium_all}e and~\ref{fig:liquid_water_equilibrium_all}g) for a $64$-bead ring polymer at a $1.4 \textrm{ fs}$ time-step with the considered choices of $\theta$.
At each time $k \Delta t$ within the $20$-picosecond interval, the non-equilibrium mean is estimated by averaging across $1000$ independent trajectories initialized at a point-mass distribution on the $n$-bead ring-polymer phase space centered at the classical (i.e., $n = 1$) sample used to initialize the reported simulations.
Within statistical uncertainty, the non-equilibrium mean for each observable converges to its equilibrium value within the $20$-picosecond interval at visually indistinguishable rates across the tested choices of $\theta$.

\begin{figure}
    \centering
    \includegraphics[width=\columnwidth]{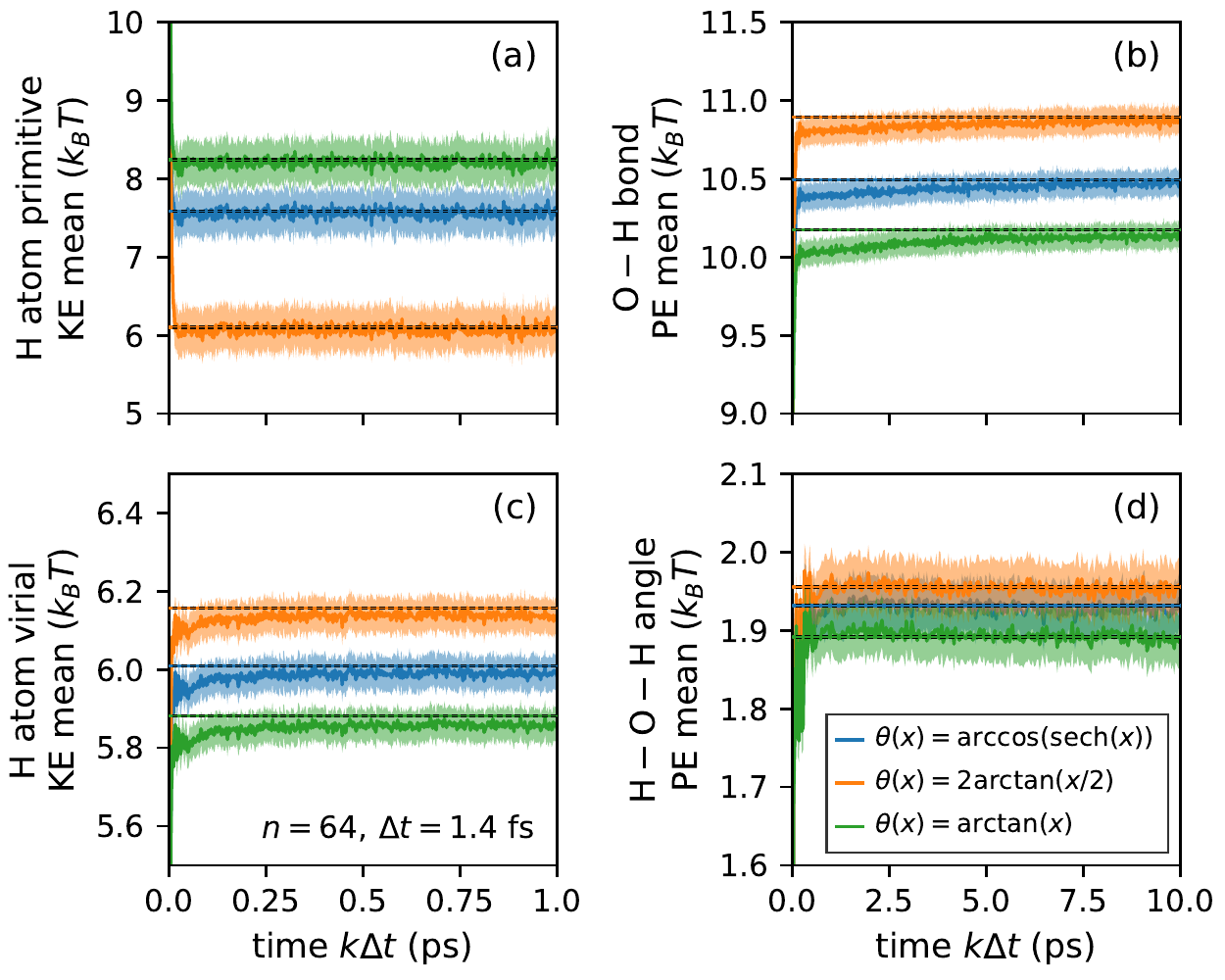}
    \caption{\small 
    Convergence to equilibrium of the BAOAB-like schemes considered in Section~\ref{liquid_water_results} with $n = 64$ ring-polymer beads and a $\Delta t = 1.4 \textrm{ fs}$ time-step.
    With respect to the non-equilibrium $64$-bead configurational distribution evolved from a point mass at a classical (i.e., $n = 1$) configuration, panels (a) and (c) plot the mean kinetic energy per $\mathrm{H}$ atom for the $n$-bead system as per the primitive and virial estimators, respectively, for times up to $1.0 \textrm{ ps}$.
    Panels (b) and (d), respectively, plot the non-equilibrium mean $\mathrm{O}\!-\!\mathrm{H}$-bond and $\mathrm{H}\!-\!\mathrm{O}\!-\!\mathrm{H}$-angle potential energy per q-TIP4P/F water molecule,\cite{Habershon2009} for times up to $10 \textrm{ ps}$.
    The lightly shaded interval around each curve corresponds to the standard error of the estimated non-equilibrium mean, computed with $1000$ bootstrap resamples from a sample of $1000$ independent trajectories.
    }
    \label{fig:liquid_water_convergence}
\end{figure}

\bibliography{main}

\end{document}